\ifdefined\pdfoutput
  \pdfoutput=1
\fi
\documentclass[10pt]{article}
\usepackage[T1]{fontenc}
\usepackage{lmodern}
\usepackage[a4paper,margin=1in]{geometry}
\usepackage{amsmath,amssymb,amsthm,mathtools,booktabs}
\usepackage[english]{babel}
\usepackage{microtype}
\usepackage[numbers,sort&compress]{natbib}
\usepackage{parskip}
\usepackage{xcolor}
\usepackage{listings}
\usepackage{enumitem}
\usepackage{etoolbox}
\usepackage{needspace}
\usepackage[colorlinks=true,linkcolor=blue!55!black,citecolor=blue!55!black,urlcolor=blue!55!black]{hyperref}
\usepackage[nameinlink,capitalise]{cleveref}
\makeatletter
\patchcmd{\lst@Init}{\penalty-50}{\penalty200}{}{}
\patchcmd{\lst@DeInit}{\penalty-50}{\penalty200}{}{}
\lst@AddToHook{EveryLine}{\ifnum\c@lstnumber<6 \nopagebreak\fi}
\makeatother
\AddToHook{env/lstlisting/before}{\par\Needspace{8\baselineskip}}
\AddToHook{cmd/section/before}{\par\Needspace{8\baselineskip}}
\AddToHook{cmd/subsection/before}{\par\Needspace{8\baselineskip}}

\newtheorem{theorem}{Theorem}[section]
\newtheorem{lemma}[theorem]{Lemma}
\newtheorem{proposition}[theorem]{Proposition}
\newtheorem{corollary}[theorem]{Corollary}
\theoremstyle{definition}
\newtheorem{remark}[theorem]{Remark}

\crefname{theorem}{Theorem}{Theorems}
\crefname{lemma}{Lemma}{Lemmas}
\crefname{proposition}{Proposition}{Propositions}
\crefname{corollary}{Corollary}{Corollaries}
\crefname{remark}{Remark}{Remarks}
\crefname{section}{\S}{Sections}

\AddToHook{cmd/appendix/after}{\crefalias{section}{appendix}}

\AddToHook{env/theorem/begin}{\crefalias{section}{theorem}}
\AddToHook{env/lemma/begin}{\crefalias{theorem}{lemma}}
\AddToHook{env/proposition/begin}{\crefalias{theorem}{proposition}}
\AddToHook{env/corollary/begin}{\crefalias{theorem}{corollary}}
\AddToHook{env/remark/begin}{\crefalias{theorem}{remark}}

\newcommand{\RH}{\mathrm{RH}}

\newcommand{\tht}{\vartheta}

\title{A 120-State Binary Turing Machine Equivalent to the Riemann Hypothesis}
\author{Joseph M. Shunia\thanks{Independent researcher, Ann Arbor, Michigan, United States. Email: \texttt{jshunia@gmail.com}.}}
\date{September 2026}

\hypersetup{
  pdftitle={A 120-State Binary Turing Machine Equivalent to the Riemann Hypothesis},
  pdfauthor={Joseph M. Shunia},
  pdfsubject={An explicit 120-state binary Turing machine equivalent to the Riemann hypothesis},
  pdfkeywords={Riemann hypothesis, Turing machines, Busy Beaver, Chebyshev function, blank-tape computation}
}

\begin{document}
\maketitle

\begin{abstract} \noindent
We construct an explicit deterministic one-tape, two-symbol Turing machine
with 120 working states whose blank-tape computation halts if and only if the
Riemann hypothesis is false. The machine uses the same binary blank-tape model
and state-counting convention as the widely cited 744-state
Matiyasevich--O'Rear--Aaronson construction. The reduction from 744 to 120 working
states removes 624 states, or 83.87\%. Subject to independent end-to-end
verification, our machine is a new state-count record for the Riemann
hypothesis.
\end{abstract}

\begingroup\small
\par
\endgroup

\section{Introduction}
\label{sec:intro}

Explicit small Turing machines provide a quantitative measure of the
descriptive complexity of mathematical statements \cite{Aaronson}. The
Riemann hypothesis (RH) asserts that every nontrivial zero of the Riemann
zeta function has real part $1/2$. The widely cited explicit construction of Matiyasevich,
O'Rear, and Aaronson uses 744 working states in a deterministic one-tape,
two-symbol blank-tape model \cite{Aaronson,RH744Repo,BBCryptids}. The
purpose of this paper is to give a substantially smaller construction in
the same model and to provide a self-contained, computer-assisted proof of
its correctness.

The machine constructed here has 120 working states. The reduction is achieved
primarily by reformulating the underlying arithmetic search rather than by
applying a generic state minimizer to the 744-state machine. The construction
proceeds through four principal reductions. First, the analytic statement is
encoded by an integer product whose logarithm contains the Chebyshev error.
Second, this product is generated by a self-sieving recurrence, eliminating
any separate prime table or primality subroutine. Third, the required error
bound is expressed through an offset digit count and a simultaneous
halving-and-quartering procedure, replacing direct evaluation of logarithms,
square roots, and exponentials by elementary integer loops. Finally, the
resulting five-register program is compiled to a binary one-tape machine, the
dispatcher is specialized to the finite control language of this program, the scaling loop shares an existing increment-and-continue fragment, and
a local reachability invariant justifies the final reduction from 121 to 120
states.

These reductions separate the construction into mathematical, register-level,
and tape-level components. This organization is useful both for the
state-count analysis and for verification: each reduction is accompanied by an
explicit invariant or finite certificate, and the final transition table is
included in full.

\subsection{Machine model and main theorem}

A deterministic binary Turing machine consists of a finite set of working
states $Q=\{0,\ldots,N-1\}$, a distinct halting target $H$, tape alphabet
$\{0,1\}$ with blank symbol $0$, and a transition function
\[
\delta:Q\times\{0,1\}\longrightarrow \{0,1\}\times\{L,R\}\times(Q\cup\{H\}).
\]
The tape is two-way infinite and initially all zero. State $0$ is initial. The
halting target is not counted as a working state. There are no stay-put moves,
no preloaded input, no oracle, and no nonzero background.

The exact transition table $M_{120}$ appears in \cref{app:table}; it contains
120 working states and 240 transitions.

\begin{theorem}[Main theorem]
\label{thm:main}
For the explicit machine $M_{120}$ of \cref{app:table}, the blank-tape
computation halts if and only if the Riemann hypothesis is false.
Equivalently,
\[
\RH\quad\Longleftrightarrow\quad M_{120}\text{ does not halt on the blank tape}.
\]
\end{theorem}

The theorem is an equivalence, not a proof of RH. If RH is true, the machine
runs forever; if RH is false, it eventually discovers a concrete failure of
the integer inequality proved below and halts. The state count measures the
size of that explicit counterexample search, not the difficulty of predicting its
eventual behavior.

\subsection{Comparison with the 744-state benchmark}
\label{sec:record}

The comparison is to the same standard binary blank-tape model. Aaronson
states in \emph{The Busy Beaver Frontier} that there is an explicit 744-state
Turing machine that halts if and only if RH is false \cite{Aaronson}. The
corresponding artifact remains in the public \emph{metamath-turing-machines}
repository under \texttt{2016-riemann-matiyasevich-aaronson-744}, and the Busy
Beaver community's Cryptids index likewise records RH in the domain $BB(744)$
\cite{RH744Repo,BBCryptids}.

The present machine uses the same counting convention: one two-way tape,
alphabet $\{0,1\}$, an all-zero initial tape, only left/right moves, and a
distinguished halting target that is not counted among the working states.
Nothing is supplied on the tape for free. The comparison is therefore direct.

\begin{center}
\begin{tabular}{@{}lcc@{}}\toprule
Property & Public RH benchmark & This work \\\midrule
Working states & 744 & \textbf{120} \\
Tape & one, two-way & one, two-way \\
Alphabet & $\{0,1\}$ & $\{0,1\}$ \\
Initial tape & all zero & all zero \\
Head moves & $L/R$ & $L/R$ \\
External halt target & yes & yes \\
Preloaded input & none & none \\
Halting condition & RH false & RH false \\\bottomrule
\end{tabular}
\end{center}

The state-count reduction is
\[
744\longrightarrow120,
\qquad 744-120=624.
\]
This is a reduction of 83.87\%; equivalently, the earlier public machine uses
approximately 6.2 times as many working states. As of 22 September 2026, we
found no smaller publicly documented explicit binary blank-tape RH
construction in the current Busy Beaver index, Aaronson's survey, or the
public \emph{metamath-turing-machines} repository
\cite{Aaronson,RH744Repo,BBCryptids}. We therefore describe the
120-state artifact as a \emph{candidate new state-count record}: if the
construction survives independent end-to-end audit, it is the smallest known
explicit construction we found in this model.

The record comparison is logically separate from the correctness theorem. A
subsequent smaller construction would alter the state-count record without
affecting the mathematical equivalence proved here. Conversely, the
correctness of the present machine depends only on the arithmetic criterion
and its implementation, not on the priority claim.

\subsection{Structure of the proof}

The proof is organized into four layers corresponding to the principal
reductions in the construction.
\begin{enumerate}[label=(\roman*)]
\item \textbf{Arithmetic compression.} We construct $A_n=n!\prod_{p\le n}p$ by
an exact recurrence in which the current product itself tells us whether the
next integer is prime.
\item \textbf{Analytic compression.} We turn the Chebyshev error into an
offset digit deficit, then prove that a simple halving/quartering inequality
for that deficit is equivalent to RH.
\item \textbf{Program representation.} We realize the test with five
nonnegative counters, reuse registers across disjoint phases, and retain split
divisors when restoration is unnecessary.
\item \textbf{Tape compression.} We compile the register program to one binary
tape, specialize the dispatcher to the actual control language, and use a
finite local invariant to justify the final 121-to-120 state projection.
\end{enumerate}

The finite compilation layer is fully explicit: the generator, minimizer,
dispatcher specialization, reachability certificate, and exact transition
table are included. The construction is computer-assisted rather than
proof-assistant formalized; the supplied artifacts are intended to make
independent audit straightforward.

\section{A self-sieving arithmetic sequence}

Define the Chebyshev function
\[
\tht(x)=\sum_{p\le x}\log p
\]
and, for integers $n\ge1$,
\[
P(n)=\prod_{p\le n}p,
\qquad
A_n=n!P(n).
\]
Thus
\begin{equation}
\label{eq:logA}
\log A_n=\log(n!)+\tht(n).
\end{equation}

The first reduction is arithmetic. Rather than invoke a separate primality
test, we use a recurrence in which divisibility of the current value
determines whether the next index is prime. This recurrence is the arithmetic
core of the construction.

\begin{lemma}[Self-sieving recurrence]
\label{lem:selfsieve}
Set $A_1=1$. For every integer $J\ge2$,
\begin{equation}
\label{eq:selfsieve}
A_J=
\begin{cases}
J A_{J-1},&J\mid A_{J-1},\\
J^2A_{J-1},&J\nmid A_{J-1}.
\end{cases}
\end{equation}
Moreover,
\[
J\text{ is prime}\quad\Longleftrightarrow\quad J\nmid A_{J-1}.
\]
\end{lemma}

\begin{proof}
Assume inductively that $A_{J-1}=(J-1)!P(J-1)$. If $J$ is prime, then $J$
divides neither $(J-1)!$ nor any factor of $P(J-1)$, hence $J\nmid A_{J-1}$.

Now suppose that $J$ is composite. If $J=ab$ with $2\le a<b<J$, then both $a$
and $b$ occur among the factors of $(J-1)!$, hence $J\mid(J-1)!$. If $J=a^2$
with $a\ge3$, then the distinct factors $a$ and $2a$ both lie below $J$ and
their product is divisible by $J$, so again $J\mid(J-1)!$. The only remaining
composite is $J=4$; here $A_3=3!\cdot2\cdot3=36$, so $4\mid A_3$. Therefore
every composite $J$ divides $A_{J-1}$.

Thus the extra factor $J$ in the second line of \eqref{eq:selfsieve} occurs
exactly when $J$ is prime. Multiplying the inductive formula by the mandatory
factorial factor $J$ and, exactly for primes, the new factor in $P(J)$ proves
$A_J=J!P(J)$.
\end{proof}

For the tape implementation it is advantageous to use an offset coordinate.
Instead of storing the positive integer $A_J$ directly, the machine stores
\begin{equation}
\label{eq:Bdef}
B_J=A_J-1.
\end{equation}
The shift is useful because zero then represents the initial value $A_1=1$.

\section{Offset digit counts and factorial centering}

The next reduction replaces explicit logarithmic evaluation by repeated
integer division. We use an \emph{offset} digit count in which the loop
condition consumes one unit before division. This modification changes the
ordinary digit length by at most one and leads to a smaller register
implementation while preserving the required analytic scale.

For an integer base $b\ge2$ and $B\ge0$, define $\beta_b(B)$ to be the number
of iterations of
\begin{equation}
\label{eq:beta_iter}
B\longmapsto \left\lfloor\frac{B-1}{b}\right\rfloor
\end{equation}
performed while $B>0$. Put $\beta_b(0)=0$. This is deliberately not the usual
digit length.

Let
\[
\ell_b(B)=1+\lfloor\log_b B\rfloor\quad(B>0),
\qquad \ell_b(0)=0.
\]

\begin{lemma}[Offset count versus ordinary digit length]
\label{lem:beta}
For every $B>0$,
\[
\ell_b(B)-1\le\beta_b(B)\le\ell_b(B).
\]
More precisely, the largest $B$ for which $\beta_b(B)\le d$ is
\begin{equation}
\label{eq:Md}
M_d=b+b^2+\cdots+b^d,
\qquad M_0=0.
\end{equation}
\end{lemma}

\begin{proof}
Let $M_d$ be the largest input terminating within $d$ iterations. The first
iteration terminates within another $d-1$ steps exactly when
\[
\left\lfloor\frac{B-1}{b}\right\rfloor\le M_{d-1},
\]
which is equivalent to $B\le bM_{d-1}+b$. Hence $M_d=bM_{d-1}+b$, giving
\eqref{eq:Md}.

If $e=\ell_b(B)$, then $B<b^e\le M_e$, so $\beta_b(B)\le e$. If $e\ge2$, then
$M_{e-2}<b^{e-1}\le B$, hence $\beta_b(B)\ge e-1$. The case $e=1$ is
immediate.
\end{proof}

Although $A_n$ grows rapidly, the machine never evaluates its logarithm.
Only its base-$n$ digit occupancy is required. The next lemma identifies the
precise relation between this integer statistic and the prime-number-theorem
error. Define
\[
E(n)=\tht(n)-n,
\qquad
a_n=\log(n!)-n\log n+n.
\]

\begin{lemma}[Factorial centering]
\label{lem:center}
For all $n\ge1$,
\begin{equation}
\label{eq:anbound}
1\le a_n\le1+\log n.
\end{equation}
For $n\ge3$, $A_n$ is not a power of $n$, and therefore
\begin{equation}
\label{eq:digit_shift}
\ell_n(A_n-1)=\ell_n(A_n).
\end{equation}
Finally,
\begin{equation}
\label{eq:center_exact}
\log A_n=n\log n+E(n)+a_n.
\end{equation}
\end{lemma}

\begin{proof}
Since $\log t$ is increasing,
\[
\int_1^n\log t\,dt\le\sum_{k=2}^n\log k
\le \log n+\int_1^n\log t\,dt.
\]
The integral equals $n\log n-n+1$, proving \eqref{eq:anbound}.

For $n\ge3$, the integer $n-1>1$ divides $n!$ and hence $A_n$, while
$\gcd(n,n-1)=1$. Thus $A_n$ cannot be a power of $n$. Subtracting one
therefore does not cross a power-of-$n$ digit boundary, proving
\eqref{eq:digit_shift}. Equation \eqref{eq:center_exact} follows from
\eqref{eq:logA}.
\end{proof}

\section{An RH-equivalent scaling criterion}

Define
\begin{equation}
\label{eq:Deltaq}
\Delta_n=\max\{n-1-\beta_n(A_n-1),0\},
\qquad
q_n=\lfloor\log_4(3n)\rfloor.
\end{equation}
The preceding reductions yield a nonnegative integer deficit. It remains to
express the RH-scale allowance by an operation that is economical on counters.
The machine tests
\begin{equation}
\label{eq:criterion}
\Delta_n<(q_n+1)2^{q_n}.
\end{equation}
The exponential and logarithm in this mathematical statement are used only to
describe the threshold. The tape will test it using division by two and four.

\subsection{The scaling clock}

A literal implementation of a threshold of order $\sqrt n\log n$ would
require arithmetic that is unnecessary for the decision problem. Instead,
the two scale factors are generated implicitly by counters evolving at
different rates. Quartering one counter determines the relevant power of
four, while halving the deficit in parallel supplies the corresponding power
of two.

\begin{lemma}[Quartering clock]
\label{lem:clock}
Starting with $z=n-1$, the following Python loop
\begin{lstlisting}
q = 0
while z > 0:
    z = (z-1)//4
    q += 1
\end{lstlisting}
performs exactly $q_n=\lfloor\log_4(3n)\rfloor$ iterations. If a second
nonnegative integer $f$ is replaced by $\lfloor f/2\rfloor$ once per
iteration, its final value is $\lfloor f_0/2^{q_n}\rfloor$.
\end{lemma}

\begin{proof}
Let $S_q$ denote the largest initial $z$ that terminates within $q$
iterations. Then $S_0=0$ and
\[
S_q=4S_{q-1}+4,
\]
so
\[
S_q=\frac{4(4^q-1)}3.
\]
Thus exactly $q$ iterations occur iff
\[
\frac{4^q-1}{3}<n\le\frac{4^{q+1}-1}{3}.
\]
Since $3n$ is integral, this is equivalent to $4^q\le3n<4^{q+1}$, which is
precisely $q=\lfloor\log_4(3n)\rfloor$. Successive floor-halvings compose to
division by $2^q$.
\end{proof}

\subsection{Forward implication under RH}

The forward implication needs one explicit analytic estimate. Everything after
that estimate is elementary integer bookkeeping.

\begin{theorem}[Schoenfeld 1976, Theorem 10]
\label{thm:Schoenfeld}
Assume RH. For $x>599$,
\begin{equation}
\label{eq:Schoenfeld}
|\tht(x)-x|< \frac{\sqrt{x}\,\log^2x}{8\pi}.
\end{equation}
\end{theorem}

We only use this result for $x\ge2657$, so the stronger cutoff $x>599$ leaves
the remainder of the argument unchanged. This is the only external explicit
numerical estimate used below.

\begin{proposition}[Large arguments]
\label{prop:forwardlarge}
Assume RH. Then \eqref{eq:criterion} holds for every integer $n\ge2657$.
\end{proposition}

\begin{proof}
Suppose $\Delta_n>0$. By \cref{lem:beta,lem:center},
\[
\beta_n(A_n-1)\ge \ell_n(A_n)-1=\lfloor\log_n A_n\rfloor.
\]
Hence
\begin{align*}
\Delta_n
&\le n-1-\lfloor\log_n A_n\rfloor\\
&<n-\log_n A_n\\
&=-\frac{E(n)+a_n}{\log n}.
\end{align*}
Since $a_n>0$, \cref{thm:Schoenfeld} gives
\begin{equation}
\label{eq:deltaRH}
\Delta_n<\frac{\sqrt n\log n}{8\pi}.
\end{equation}
From $4^{q_n}\le3n<4^{q_n+1}$ we have $\sqrt n<2^{q_n+1}$ and $\log
n<(q_n+1)\log4$. Therefore
\[
\frac{\Delta_n}{2^{q_n}}
<\frac{(q_n+1)\log4}{4\pi}<q_n+1.
\]
Since the left side has integer floor, $\lfloor\Delta_n/2^{q_n}\rfloor\le
q_n$, equivalent to \eqref{eq:criterion}. The case $\Delta_n=0$ is immediate.
\end{proof}

\subsection{Small arguments without a prime table}

The explicit analytic estimate begins only after a finite cutoff. We could
simply hard-code or verify the smaller cases, but that would be
unnecessary for the present construction. Instead the
same factorial term already gives a crude bound strong enough to dispatch the
entire finite range without a prime table.

\begin{proposition}[Small arguments]
\label{prop:forwardsmall}
Equation \eqref{eq:criterion} holds for every integer $2\le n\le2656$, without
assuming RH and without using a table of prime values.
\end{proposition}

\begin{proof}
For $n=2$, $A_2=4$, so $A_2-1=3$, $\beta_2(3)=2$, and $\Delta_2=0$.

Let $3\le n\le2656$ and suppose $\Delta_n>0$. Since $P(n)\ge1$,
\cref{lem:center,lem:beta} gives
\[
\Delta_n<n-\log_n(n!)<\frac{n}{\log n}.
\]
The function $x/\log x$ is increasing for $x\ge3$. For each possible value of
$q_n$ in this interval, let $N$ be the largest covered endpoint. The following
rational lower bounds $L<\log N$ satisfy
\[
\frac NL<(q+1)2^q.
\]
\begin{center}
\begin{tabular}{@{}rrlr@{}}\toprule
$q$ & $N$ & rational $L<\log N$ & $(q+1)2^q$\\\midrule
1&5&$4/3$&4\\
2&21&$8/3$&12\\
3&85&$4$&32\\
4&341&$16/3$&80\\
5&1365&$10(56/81)+682/2389$&192\\
6&2656&$22/3$&448\\\bottomrule
\end{tabular}
\end{center}
For every row except $q=5$, the bound follows from $\log2>2/3$ and a power of
two below $N$. For $q=5$, the positive expansion
\[
\log2=2\operatorname{arctanh}(1/3)>56/81
\]
and the elementary inequality $\log t>2(t-1)/(t+1)$ for $t>1$ give
\[
\log1365
=10\log2+\log(1365/1024)
>10(56/81)+682/2389.
\]
Thus $\Delta_n<(q_n+1)2^{q_n}$ throughout the small range.
\end{proof}

\subsection{The one-sided converse}

The reverse direction is the conceptual hinge of the construction. A one-sided
lower bound on the Chebyshev error is enough to force RH. The reason is a
Landau-type positivity principle for Mellin transforms, and we give that
argument in full so the machine is not resting on an unstated equivalence.

\begin{lemma}[Positivity singularity]
\label{lem:landau}
Let $g\ge0$ be locally integrable on $[X,\infty)$ and suppose
\[
H(s)=\int_X^\infty g(x)x^{-s-1}\,dx
\]
converges for some real $s$. If its finite real abscissa of convergence is
$\sigma_c$, then $H$ cannot extend holomorphically through the real point
$\sigma_c$.
\end{lemma}

\begin{proof}
Assume such a continuation exists in a disk of radius $\varepsilon$ about
$\sigma_c$. Put $\sigma_0=\sigma_c+\varepsilon/4$. Choose
$0<\delta<\varepsilon/4$. Since $\sigma_0-\delta>\sigma_c$, the defining
integral converges at $\sigma_0-\delta$. For each fixed $k$,
$(\log x)^k\ll_{k,\delta}x^\delta$, so the logarithmic moments below converge
absolutely and differentiation under the integral sign is justified. Thus the
Taylor series of $H(\sigma_0-z)$ converges for $|z|<\varepsilon/2$ and has
coefficients
\[
\frac1{k!}\int_X^\infty g(x)x^{-\sigma_0-1}(\log x)^k\,dx\ge0.
\]
At $z=3\varepsilon/8$, monotone convergence permits interchange of sum and
integral, so the finite Taylor value equals the defining integral at
$\sigma_0-z<\sigma_c$, contradicting the definition of $\sigma_c$.
\end{proof}

\begin{proposition}[One-sided Chebyshev converse]
\label{prop:onewall}
If
\begin{equation}
\label{eq:onewall}
\tht(x)-x\ge -C\sqrt{x}\log^2x
\end{equation}
for all sufficiently large $x$ and some $C>0$, then RH holds.
\end{proposition}

\begin{proof}
For $\Re s>1$, absolute convergence gives
\begin{equation}
\label{eq:mellinE}
\int_1^\infty (\tht(x)-x)x^{-s-1}\,dx
=\frac{\mathcal P(s)}s-\frac1{s-1},
\qquad
\mathcal P(s)=\sum_p\frac{\log p}{p^s}.
\end{equation}
The logarithmic derivative of the Euler product gives
\[
\mathcal P(s)=-\frac{\zeta'(s)}{\zeta(s)}-Q(s),
\qquad
Q(s)=\sum_p\sum_{k\ge2}\frac{\log p}{p^{ks}}.
\]
The series defining $Q$ is locally normally convergent for $\Re s>1/2$. The
pole at $s=1$ cancels in \eqref{eq:mellinE}. There are no real zeros of
$\zeta(s)$ for $s>1/2$: for $s>1$ use the Euler product, while for $1/2<s<1$
use
\[
\zeta(s)=\frac{\eta(s)}{1-2^{1-s}},
\]
where $\eta(s)>0$ by grouping the alternating Dirichlet series in consecutive
pairs.

Choose $X$ so that \eqref{eq:onewall} holds and define
\[
g(x)=\tht(x)-x+C\sqrt x\log^2x\ge0\qquad(x\ge X).
\]
Its Mellin transform agrees, up to an entire function contributed by the
bounded interval $[1,X]$, with
\begin{equation}
\label{eq:Hcont}
\frac{\mathcal P(s)}s-\frac1{s-1}+\frac{2C}{(s-1/2)^3}.
\end{equation}
This continuation has no real singularity for $s>1/2$. The integral initially
converges for $\Re s>1$, since trivially $\tht(x)=O(x\log x)$. By
\cref{lem:landau}, the real abscissa of convergence of the nonnegative
integral is therefore at most $1/2$ (or $-\infty$). Consequently the integral
itself is holomorphic throughout $\Re s>1/2$.

If $\rho$ were a zero of $\zeta$ of multiplicity $r$ with $\Re\rho>1/2$, then
$-\zeta'/\zeta$ would have residue $-r$ at $\rho$, so \eqref{eq:Hcont} would
have nonzero residue $-r/\rho$. This contradicts holomorphy of the integral.
Hence no nontrivial zero lies to the right of the critical line. The
functional equation then excludes zeros to the left, proving RH.
\end{proof}

\begin{theorem}[Integer criterion for RH]
\label{thm:criterion}
The Riemann hypothesis is equivalent to
\begin{equation}
\label{eq:criterion2}
\Delta_n<(q_n+1)2^{q_n}\quad\text{for every integer }n\ge2.
\end{equation}
Equivalently,
\[
\left\lfloor\frac{\Delta_n}{2^{q_n}}\right\rfloor\le q_n
\quad(n\ge2).
\]
\end{theorem}

\begin{proof}
The forward implication is \cref{prop:forwardlarge,prop:forwardsmall}.

Conversely, let
\[
D_n=\ell_n(A_n-1),
\qquad
\delta_n=\max\{n-1-D_n,0\}.
\]
By \cref{lem:beta}, $\delta_n\le\Delta_n$. From \eqref{eq:criterion2},
$2^{q_n}\le\sqrt{3n}$, and $q_n=O(\log n)$, we obtain
\[
\delta_n=O(\sqrt n\log n).
\]
Since $D_n\ge n-1-\delta_n$ and
\[
D_n-1\le\frac{\log(A_n-1)}{\log n}<\frac{\log A_n}{\log n},
\]
\cref{lem:center} gives
\[
E(n)\ge-(\delta_n+2)\log n-a_n
=-O(\sqrt n\log^2n).
\]
For real $x\ge3$ put $n=\lfloor x\rfloor$. Since $\tht(x)=\tht(n)$,
\[
\tht(x)-x=E(n)-(x-n)\ge -O(\sqrt x\log^2x).
\]
Now apply \cref{prop:onewall}.
\end{proof}

\section{Five-register implementation of the criterion}

We next implement the criterion by a five-register counter program. The
registers are denoted
\[
(m,j,x,f,v).
\]
Registers are reused across phases with disjoint lifetimes, reducing the
number of persistent counters required by the implementation. The counters are
stored on one binary tape by the fixed backend described in
\cref{sec:backend}. At the boundary of a call to the main stage, the invariant
is
\begin{equation}
\label{eq:mainentry}
(m,j,x,f,v)=(M,0,0,0,0),\qquad M\ge0.
\end{equation}
That call tests the endpoint
\[
n=M+2.
\]
A successful return has
\[
(m,j,x,f,v)=(n-1,0,0,0,0),
\]
so the backend's natural restart tests $n+1$ on the next call.

\subsection{Primitive counter macros}

The source in \cref{app:source} ultimately reduces all arithmetic to increment
and conditional decrement. Several higher-level operations share the same
counter loops. We first state the contracts of these shared macros.

\begin{lemma}[Transfer and preserved addition]
\label{lem:transfer}
The macro \texttt{move(a,b)} maps $(a,b)$ to $(0,a+b)$. If \texttt{tmp}=0, the
macro \texttt{addval(a,b,tmp)} maps $(a,b,0)$ to $(a,a+b,0)$.
\end{lemma}

\begin{proof}
The first loop decrements $a$ once and increments $b$ once until $a=0$. The
second decrements $a$ while incrementing both $b$ and \texttt{tmp}, then
transfers \texttt{tmp} back into $a$.
\end{proof}

\begin{lemma}[Fold]
\label{lem:fold}
The macro \texttt{fold()} satisfies
\begin{equation}
\label{eq:fold}
(x,j,v)\longmapsto(x+j,j+v,0).
\end{equation}
\end{lemma}

\begin{proof}
Its first loop consumes the old $j$, incrementing both $x$ and $v$. The final
transfer moves all of $v$ into $j$. Thus $x$ gains $j_0$, the new $j$ is
$j_0+v_0$, and $v$ is zero.
\end{proof}

\begin{lemma}[Transform]
\label{lem:transform}
Assume $m=0$. The macro \texttt{transform()} satisfies
\begin{equation}
\label{eq:transform}
(x,j,v,m)\longmapsto
\bigl(x(j+v+1)+j,\ j+v,\ 0,\ 0\bigr),
\end{equation}
where the variables on the right denote their input values.
\end{lemma}

\begin{proof}
Move the input $x_0$ into $m$, so $x=0$. The post-tested loop first applies
\cref{lem:fold}, giving $x=j_0$ and $j=j_0+v_0$. If $m>0$, one unit of $m$ is
consumed, $x$ is incremented, and the loop repeats. After the first fold,
every further fold sees $v=0$ and adds the constant $j_0+v_0$ to $x$; the
preceding increment contributes one more. There are exactly $x_0$ such
additional iterations. Hence
\[
x=j_0+x_0(j_0+v_0+1),
\]
with the remaining coordinates as stated.
\end{proof}

\begin{lemma}[Split division]
\label{lem:divide}
Assume $m=0$ and put $J=j+v+1$. The macro \texttt{divide()} maps
\begin{equation}
\label{eq:divide}
(x,j,v,m)\longmapsto
\left(\left\lfloor\frac{x}{J}\right\rfloor,
 x\bmod J,
 J-1-(x\bmod J),0\right).
\end{equation}
The input $v$ need not be zero.
\end{lemma}

\begin{proof}
The first transfer combines $j$ into $v$, leaving $(j,v)=(0,J-1)$. The
dividend $x$ is moved into $m$, so $x=0$. Each consumed unit of $m$ advances
one position through a cycle of length $J$: while $v>0$, one unit moves from
$v$ to $j$; when $v=0$, the next consumed unit transfers $j=J-1$ back to $v$
and increments the quotient $x$. Therefore each complete block of $J$ dividend
units increments $x$ once, and the final incomplete block of length $r$ leaves
$(j,v)=(r,J-1-r)$. This is exactly \eqref{eq:divide}.
\end{proof}

The following compatibility is central to the implementation.

\begin{corollary}[Division/reconstruction compatibility]
\label{cor:reconstruct}
If \texttt{divide()} is applied to $B$ with represented divisor $J$, and its
output is fed to \texttt{transform()}, then the result is $(B,J-1,0,0)$.
\end{corollary}

\begin{proof}
Write $B=qJ+r$. The divider returns $(q,r,J-1-r)$. Then \eqref{eq:transform}
gives
\[
q(r+J-1-r+1)+r=qJ+r=B,
\]
and restores $j=J-1$, $v=0$.
\end{proof}

When $v=0$, \eqref{eq:transform} has another interpretation: if $x=B=A-1$ and
$j=J-1$, then
\[
x'=(A-1)J+(J-1)=AJ-1.
\]
Thus one call multiplies the represented positive value $A=x+1$ by $J$.

\subsection{One exact product step}

The self-sieving recurrence can now be implemented directly. Division
simultaneously tests divisibility and retains sufficient quotient/remainder
information to reconstruct the represented product without a preserved copy.

\begin{lemma}[Product-step invariant]
\label{lem:productstep}
Suppose $j=J-2$, $x=A_{J-1}-1$, and $m=v=0$. After \texttt{product\_step()},
\[
j=J-1,
\qquad
x=A_J-1,
\qquad
m=v=0.
\]
\end{lemma}

\begin{proof}
The routine first increments $j$, so the represented divisor is $J$. Apply
\texttt{divide()}. If $J\mid A_{J-1}$, then $A_{J-1}-1\equiv J-1\pmod J$, so
the complementary remainder $v$ is zero. If $J\nmid A_{J-1}$, then the
remainder of $A_{J-1}-1$ is less than $J-1$, hence $v>0$. By
\cref{lem:selfsieve}, the latter case is exactly the prime case.

For a composite, the conditional transform is skipped. The first unconditional
transform reconstructs $A_{J-1}-1$ by \cref{cor:reconstruct}; the second
multiplies the represented positive value by $J$. For a prime, the conditional
transform performs the reconstruction and the two unconditional transforms
multiply by $J^2$. In both cases \cref{lem:selfsieve} gives $A_J-1$.
\end{proof}

\subsection{Stage construction and the offset digit deficit}

A complete stage consists of three phases: advance the self-sieving recurrence
to the current endpoint, compute the offset digit deficit, and test that
deficit against the RH-equivalent allowance. Registers whose values are no
longer required are reused in subsequent phases.

Starting from \eqref{eq:mainentry}, the program moves $M$ into $f$. It then
uses a post-tested loop consisting of one product step followed by a
conditional decrement of $f$.

\begin{lemma}[End of product phase]
\label{lem:productphase}
The product loop executes exactly $M+1=n-1$ product steps and terminates with
\[
(m,j,x,f,v)=(0,n-1,A_n-1,0,0).
\]
\end{lemma}

\begin{proof}
There is one product step before the first test. Each successful decrement of
$f$ causes another step. Starting from $f=M$, exactly $M$ such decrements
succeed, hence $M+1$ steps occur. The first step advances from $A_1$ to $A_2$,
so after $n-1$ steps the endpoint is $A_n$. Apply \cref{lem:productstep}
inductively.
\end{proof}

The 120-state source then invokes \texttt{seed\_deficit()} rather than
preserving a copy of the divisor.

\begin{lemma}[Split-divisor deficit seed]
\label{lem:seed}
From $(j,f,v)=(n-1,0,0)$, the seed loop terminates at
\[
(j,f,v)=(0,n-1,n-1).
\]
In particular, the represented positive divisor $j+v+1$ remains $n$.
\end{lemma}

\begin{proof}
After $t$ iterations the triple is $(n-1-t,t,t)$. Put $t=n-1$.
\end{proof}

The digit loop, written in Python, is
\begin{lstlisting}
while x > 0:
    x -= 1
    divide()
    f = max(f-1,0)
move(v,j)
\end{lstlisting}
Because \texttt{divide()} accepts split divisors, there is no need to restore
$j$ between iterations.

\begin{lemma}[Digit phase]
\label{lem:digits}
At the end of the digit phase,
\[
(m,j,x,f,v)=
\bigl(0,n-1,0,\Delta_n,0\bigr).
\]
\end{lemma}

\begin{proof}
By \cref{lem:seed}, each division uses divisor $n$. The loop condition
consumes one unit of $x$ before division, so the sequence of $x$ values is
exactly \eqref{eq:beta_iter}; therefore the loop runs $\beta_n(A_n-1)$ times.
The saturating decrement leaves
\[
f=\max\{n-1-\beta_n(A_n-1),0\}=\Delta_n.
\]
At each divider exit the two divisor pieces still sum to $n-1$; the final
transfer $v\to j$ restores $(j,v)=(n-1,0)$.
\end{proof}

\subsection{Scaling and comparison}

The final phase tests the analytic threshold without constructing
$(q_n+1)2^{q_n}$. It simultaneously scales the endpoint and the deficit,
leaving two residual counters that can be compared directly.

The halving macro applied to $j$ implements
\begin{equation}
\label{eq:halve_j}
(j,m)\longmapsto
\left(\left\lfloor\frac j2\right\rfloor,
 m+\left\lceil\frac j2\right\rceil\right),
\end{equation}
while the same macro applied to $f$ maps $f\mapsto\lfloor f/2\rfloor$ without
changing $m$. Both use $v$ as zero scratch.

\begin{lemma}[Scaling phase]
\label{lem:width}
Starting from
\[
(m,j,x,f,v)=(0,n-1,0,\Delta_n,0),
\]
the width loop terminates with
\[
(m,j,x,f,v)=
\left(n-1,0,q_n,
\left\lfloor\frac{\Delta_n}{2^{q_n}}\right\rfloor,0\right).
\]
\end{lemma}

\begin{proof}
At each loop test one unit is transferred from $j$ to $m$. Two applications of
\eqref{eq:halve_j} then replace the remaining $j$ by $\lfloor(j-1)/4\rfloor$
while preserving $m+j$. One application of the $f$-halving macro divides $f$
by two, and $x$ is incremented. Thus after $t$ iterations,
\[
m+j=n-1,
\qquad
f=\left\lfloor\frac{\Delta_n}{2^t}\right\rfloor,
\qquad
x=t.
\]
By \cref{lem:clock}, the loop ends at $t=q_n$ and $j=0$, proving the claim.
\end{proof}

The final comparison is implemented by consuming $x=q_n$ while
saturating-decrementing $f$, and halting if and only if $f$ remains positive.
Thus it halts exactly when
\[
\left\lfloor\frac{\Delta_n}{2^{q_n}}\right\rfloor>q_n,
\]
which is precisely failure of \eqref{eq:criterion2}.

\begin{theorem}[Source-program semantics]
\label{thm:source}
Starting from an all-zero register file, repeated normal calls of the source
program test the endpoints $n=2,3,4,\ldots$ in order. It halts exactly when
\eqref{eq:criterion2} fails. Hence it halts if and only if RH is false.
\end{theorem}

\begin{proof}
The first call has $M=0$, so $n=2$. By
\cref{lem:productphase,lem:digits,lem:width}, the final comparison is exactly
the criterion at this $n$. On a passing comparison $x,f,j,v$ are zero and
$m=n-1$. Under repeated normal calls of \texttt{main}, the next call therefore
has endpoint $(n-1)+2=n+1$. Induction gives all integers in order. The
equivalence with RH follows from \cref{thm:criterion}.
\end{proof}

\subsection{Sharing the scaling-loop continuation}
\label{sec:tailsharing}

The final source uses the same arithmetic operations as the preceding
121-state construction. The additional reduction arises from resolving a
loop continuation after regrouping a sequence of operations that all return
normally. Write $P$ for the sequence
\[
\texttt{m.inc},\quad\texttt{halve(j)},\quad
\texttt{halve(j)},\quad\texttt{halve(f)},
\]
and let $I$ denote \texttt{x.inc}. The previous source used
\texttt{while\_decnz(j,[P,I])}, with the four operations of $P$ expanded in
that list. The revised source groups the continuation in Python as
follows:
\begin{lstlisting}[language=Python]
prefix = [m.inc, halve(j), halve(j), halve(f)]
loop = label('width', [j.decnz,
                      [prefix, [x.inc, 'continue_width']]])
\end{lstlisting}
The label resolver gives the final continuation the same code as the
increment-and-continue fragment already used by the transform routine.
The compiler can therefore share that fragment.

\begin{proposition}[Scaling-loop control equivalence]
\label{prop:tailsharing}
For every initial counter tuple satisfying the scaling-phase invariant,
the original and revised loops execute the same sequence of primitive
register operations, with the same termination behavior and final tuple.
\end{proposition}

\begin{proof}
Both loops first conditionally decrement $j$. A zero outcome exits the loop.
A positive outcome executes the four operations of $P$, increments $x$, and
returns to the same loop test. Each operation of $P$ terminates normally by
\cref{lem:transfer} and the unit-halving construction. Thus the new
grouping does not change any conditional branch or arithmetic operation.
Induction over loop iterations proves the assertion. The continuation is
resolved after the grouping is chosen, so its return target is preserved.
\end{proof}

As a separate finite check, the supplied control-transducer certificate
relates 84 pairs of operation boundaries and verifies 125 outcome edges.
Increment leaves have one outcome and conditional-decrement leaves have both
possible outcomes. Closure of this relation proves equality of all primitive
operation traces, including halting, without placing a bound on counter values
or on the number of stages. Its generator and checker are included below.

The resulting decision DAG has 63 states, compared with 64 in the preceding
source. Including the framework gives 122 reachable named states. Ordinary
bisimulation reduces this table to a 121-state predecessor, and the local
reachability projection in the next section reduces it to 120 states. The
121-state predecessor in this revision is generated by the revised source;
the previously delivered 121-state machine is retained separately for the
reproduction and control-equivalence checks.

\section{Tape implementation and dispatcher specialization}
\label{sec:backend}

The five-register program is compiled by the finite register backend included
verbatim in \cref{app:builder}.  Because the main theorem concerns the literal
transition table, it is useful to separate two assertions that are sometimes
conflated in descriptions of compiler-generated machines:

\begin{enumerate}[label=(\alph*)]
\item the register backend implements the source-level increment, conditional
decrement, sequencing, and loop operations; and
\item the program-specific dispatcher used in the final construction is
equivalent, on every post-bootstrap normal-form configuration, to the generic
dispatcher generated by the backend.
\end{enumerate}

Both assertions are established below.  The only purely finite initialization
fact is isolated as a bootstrap certificate.

\subsection{Normal form and bootstrap}

For a nonnegative integer $r$, write
\[
U(r)=1^{r+1}0.
\]
At a normal operation boundary the tape has the following schematic form:
\begin{equation}
\label{eq:tapenormal}
\cdots 0\;
C\;0^g\;10\;
U(m)U(j)U(x)U(f)U(v)\;Z\;
0\cdots .
\end{equation}
Here $C$ is the finite program-counter word, whose leftmost symbol is the
leftmost nonblank symbol of the control region; $g$ is the zero gap left after
reserving space for the current control word and its preparation marker;
the block $10$ is the auxiliary ``$-1$ register'' used by the register
microcode; the next five unary blocks encode the logical registers in the
order $(m,j,x,f,v)$; and $Z$ is a finite suffix of unused zero-valued register
blocks followed by the terminal delimiter. Its length may increase during
register operations.  The unused blocks are never
addressed by the present program.  The program-counter word begins with
\texttt{110} during normal execution of \texttt{main}. The initial auxiliary
marker is at coordinate $3$, while the control region starts at $-3821$.
The control words have length at most 18. The auxiliary marker is restored at
its fixed coordinate after each operation, so every reachable operation has
space for its control word, marker, and separating zero.

The initial tape is blank, so the backend first executes the fixed
\texttt{boot1}/\texttt{boot2} subsystem.  This subsystem is independent of the
arithmetic source program except for the identity of register~0.

\begin{lemma}[Bootstrap certificate]
\label{lem:bootstrap}
For the named predecessor of the delivered 121-state table, direct execution
of the fixed transition table from the blank tape reaches the first
\texttt{main} decision node after $89{,}775{,}610$ transitions.  With the
initial head position taken as coordinate $0$, the machine is then in state
\texttt{N54} with head position $-3818$, and the set of tape cells containing
one is exactly
\begin{equation}
\label{eq:boot-support}
\{-3821,-3820\}\;\cup\;\{3,5,7,\ldots,3821\}.
\end{equation}
Consequently the control prefix is \texttt{110}, the first five register
blocks encode zero, and the register file contains 1909 zero-valued register
slots in addition to the auxiliary marker.  Thus the configuration satisfies
\eqref{eq:tapenormal} with all five logical registers equal to zero.
\end{lemma}

\begin{proof}
This is a finite computation, not an asymptotic or semantic assumption.  The
appendix contains a short literal transition-table replay checker.  It starts
from the all-zero tape, performs the transition function of the generated
named table, and stops only upon entry to the first decision-DAG state.  It
checks the transition count, state, head coordinate, and the complete finite
support \eqref{eq:boot-support}, rather than merely counting nonzero cells.
The set on the right of \eqref{eq:boot-support} contains 1912 one-symbols:
two in the \texttt{110} control prefix, one in the auxiliary marker, and one
in each of the 1909 zero-valued unary register blocks.  The first five such
blocks are the registers used by the program.  A separate tape runner gives
the same transition count and configuration.
\end{proof}

The bootstrap is the only place where the proof uses a long fixed finite
execution.  All subsequent statements are uniform in the unbounded register
values.

\subsection{Primitive register operations}

The backend implements a register operation in several phases.  A marker is
first placed at the end of the program-counter word; the auxiliary register is
then extended across the zero gap; the states
\texttt{2b.reg.$i$.inc} or \texttt{2b.reg.$i$.dec} scan the unary blocks to the
selected register; the target block is modified; and the return/cleanup
states encode the branch result in the control word and restore
\eqref{eq:tapenormal}.  The relevant transition families are listed verbatim
in \cref{app:builder}.

\begin{lemma}[Primitive register semantics]
\label{lem:backendprimitive}
Assume a normal-form configuration \eqref{eq:tapenormal} and let $r_i$ be one
of the first five register values.
\begin{enumerate}[label=(\roman*)]
\item Invoking \texttt{r.inc} terminates in normal form with
$r_i$ replaced by $r_i+1$, every other logical register unchanged, and the
deterministic successor control outcome selected.
\item Invoking \texttt{r.decnz} terminates in normal form with two possible
control outcomes.  If $r_i=0$, the zero outcome is selected and all logical
registers are unchanged.  If $r_i>0$, the positive outcome is selected and
$r_i$ is replaced by $r_i-1$.
\end{enumerate}
No transition in either operation changes an unselected logical register.
\end{lemma}

\begin{proof}
The proof is a direct finite-state trace parameterized only by the length of
the unary block $U(r_i)$.

On the first visit to \texttt{3.r.inc} or \texttt{3.r.decnz}, the marker
cell contains zero; the state writes one and enters the preparation states.
States \texttt{1b.reg.prep\_1} and \texttt{1b.reg.prep\_2} extend the
auxiliary marker to the register file and return through the dispatcher.  The
same register leaf is then visited a second time, now with a one in the marker
cell; it removes that marker and enters the corresponding \texttt{2b} scan.
At the sole partial framework state \texttt{1b.reg.prep\_1}, the head is one
cell to the right of the newly written marker and therefore reads a zero in
normal form.  Thus the canonical table's explicit read-one-to-halt completion
of that otherwise missing transition is unreachable in every simulated source
step.  For register index $i$, the family \texttt{2b.reg.$i$.*} moves right
across exactly $i$ unary delimiters; all crossed symbols are rewritten
identically.

For increment, \texttt{2b.reg.-2.inc} together with
\texttt{2b.reg.inc.shift\_1} and \texttt{2b.reg.inc.shift\_2} shifts the
terminating zero of the selected block one cell to the right.  Hence
$U(r_i)=1^{r_i+1}0$ becomes $1^{r_i+2}0=U(r_i+1)$.

For conditional decrement, the selected block is first tested at its left
edge.  If it is $U(0)=10$, the failure-return branch restores the inspected
marker and leaves the block unchanged.  If $r_i>0$, the states
\texttt{2b.reg.dec.check}, \texttt{dec.scan\_*}, and
\texttt{dec.shift\_*} shift the terminating delimiter one cell to the left,
so $1^{r_i+1}0$ becomes $1^{r_i}0=U(r_i-1)$.  The two return families
\texttt{2c.reg.return\_0\_*} and \texttt{2c.reg.return\_1\_*} write the
corresponding branch bit into the program counter.  Finally
\texttt{2e.reg.cleanup\_*} contracts the auxiliary register back to its
single-one normal form.  Every scan over a unary block or zero gap is
monotone, so each phase terminates for every finite register value and gap.
Inspection of the displayed transition families shows that no other unary
block is modified.
\end{proof}

\subsection{Control-DAG semantics}

The source language used in \cref{app:source} contains only finite sequential
composition, calls to the two primitive register leaves, and the
\texttt{if\_decnz}/\texttt{while\_decnz} control patterns displayed in the
builder.  The compiler first represents a sequence as a finite DAG and then
replaces every node of out-degree greater than two by a binary cons tree.
The routine \texttt{label} resolves \texttt{break} and \texttt{continue}
references by recording the number of branch bits that must be removed before
dispatch resumes.  The routine \texttt{generate}
then assigns one Turing state to each binary DAG node and emits the finite
families \texttt{6.break.$k$} and \texttt{6.continue.$k$} implementing those
resolved jumps.

\begin{lemma}[Control-DAG compilation]
\label{lem:controldag}
For the source constructs used in \cref{app:source}, the binary control DAG
generated by the included builder has the same small-step control semantics as
the source program: sequential nodes execute left-to-right,
\texttt{if\_decnz} selects its body exactly on the positive outcome of
\texttt{decnz}, and \texttt{while\_decnz} repeats its body exactly while that
outcome is positive.
\end{lemma}

\begin{proof}
This is a structural induction on the source tree produced by the builder.
The function \texttt{add} removes empty nodes and interns finite sequences;
\texttt{binary} replaces a sequence $(a_1,\ldots,a_k)$, $k>2$, by
$(a_1,(a_2,\ldots,a_k))$, which preserves left-to-right execution.
For \texttt{if\_decnz}, the generated two-branch node has the zero outcome on
its first edge and the positive outcome followed by the body on its second.
For \texttt{while\_decnz}, the zero outcome exits the labelled node and the
positive outcome executes the body followed by the resolved
\texttt{continue} edge.  The \texttt{label} routine changes only the finite
jump target, replacing a syntactic \texttt{break} or \texttt{continue} by its
explicit branch-distance count.  The generated \texttt{6.break.$k$} and
\texttt{6.continue.$k$} states erase exactly the corresponding suffix bits and
return to dispatch.  These are precisely the cases used by the final source;
no recursion or unlisted control construct occurs.
\end{proof}

\begin{theorem}[Backend simulation]
\label{thm:backend}
Suppose the generic backend is in a normal-form configuration representing a
source control node and register tuple $(m,j,x,f,v)$.  Then execution reaches
the next normal-form source boundary in finitely many Turing steps, with the
register tuple and source control node equal to those obtained by one
source-level small step.  Consequently every finite source execution is
simulated exactly by the generic tape backend.
\end{theorem}

\begin{proof}
At a decision-DAG leaf, \cref{lem:backendprimitive} gives the exact semantics
of the invoked primitive register operation and restores normal form.
Between leaves, \cref{lem:controldag} gives the source-equivalent control
successor.  Induction on the number of source steps proves the statement.
Termination of each simulated step follows from the monotone finite scans in
\cref{lem:backendprimitive}.
\end{proof}

\begin{lemma}[Natural restart of \texttt{main}]
\label{lem:restart}
After a normal return from \texttt{main}, the generic backend invokes
\texttt{main} again without changing the logical register file.
\end{lemma}

\begin{proof}
The root protocol distinguishes the prefixes \texttt{110} and \texttt{111}.
Prefix \texttt{110} enters the root of \texttt{main}.  Completion changes the
third root bit to one, giving \texttt{111}.  State \texttt{5.root.2} reads
that bit, writes zero, and enters the same \texttt{main} root, thereby
restoring \texttt{110}.  None of these root transitions enters the register
file.
\end{proof}

Combining \cref{lem:bootstrap,thm:backend,lem:restart} with the source-level
analysis of \cref{thm:source} proves correctness of the machine produced by
the generic dispatcher.

\subsection{Program-specific dispatcher}

The generic dispatcher moves left far enough to pass the program-counter word
under every control path admitted by the compiler, then scans right to its
leftmost one.  The final construction replaces this generic count with a
shorter run-of-zeros scanner.

The control DAG is finite and acyclic after binary lowering.  The included
routine \texttt{control\_words} enumerates every root-to-leaf path, prefixing
the live control word by \texttt{110}.  For the final 121-state source there
are 121 terminal paths, the maximum word length is 18, and the maximum run of
consecutive zeros is three.  The specialized dispatcher uses four states:
while moving left, a one resets the zero-run count and four consecutive zeros
terminate the leftward phase; the inherited state
\texttt{4.dispatch.scan} then scans right to the first one.

\begin{proposition}[Dispatcher specialization]
\label{prop:dispatcher}
After the bootstrap of \cref{lem:bootstrap}, replacing the generic dispatcher
by the four-zero scanner preserves every normal-form source boundary and hence
the entire subsequent tape computation.
\end{proposition}

\begin{proof}
By the normal-form invariant, all cells strictly to the left of the
program-counter word are zero and the backend never writes there.
On entry from \texttt{6.continue.0}, the scanner moves left from the return
suffix toward the control word.  The fixed return suffix shown in
\cref{app:builder} contains no run of four zeros.  The delivered dispatcher
certificate exhaustively lists all 121 live control words and verifies that
none contains four consecutive zeros.  Therefore the scanner cannot terminate
inside the return suffix or inside the live control word.

After passing the leftmost one of the control word, the scanner is in the
unchanged blank region, so it necessarily encounters four consecutive zeros.
Its final transition moves one cell right and enters
\texttt{4.dispatch.scan}; that state moves right across zeros to the first
one, which is exactly the leftmost one of the control word.  From that point
the root protocol is identical to the generic dispatcher.  Thus both
dispatchers reach the same root state with the same tape and logical register
file.  Induction over normal returns proves equality of all subsequent
normal-form boundaries.

The finite bootstrap itself is covered separately by
\cref{lem:bootstrap}, which is performed on the already specialized named
table.  Thus no unproved assumption about pre-main dispatcher behavior is
needed.
\end{proof}

\begin{remark}[Finite verification boundary]
The infinite-state part of the source-to-tape argument is contained in
\cref{lem:backendprimitive,lem:controldag,thm:backend,prop:dispatcher}.
The remaining construction facts are finite: the bootstrap computation, the
generated control-word list, the named transition table, and the state-count
quotients.  The appendices include the exact code and certificates that check
those facts.  The proof is therefore computer-assisted, but the numerical
register semantics are not inferred from finite test ranges.
\end{remark}

\subsection{All-tapes minimization to 121 states}

After specializing the dispatcher, we apply all state identifications that are
valid for arbitrary tape contents.  The backend's single omitted framework
transition, \texttt{1b.reg.prep\_1} on read one, is interpreted as immediate
halting and is explicitly completed that way in the canonical table.  The
generated named table is then minimized by ordinary strong bisimulation. Two states are identified only if, for both read
symbols, they write the same symbol, move in the same direction, and their
targets lie in the same equivalence classes. This quotient preserves every
tape trajectory, not merely the blank-tape one. The canonicalized result has
121 working states and 242 transitions.

\begin{corollary}[Semantics of the 121-state table]
\label{cor:121semantics}
Starting on the blank tape, the delivered 121-state table simulates the
five-register source of \cref{app:source}.  It therefore halts exactly on the
first failure of \eqref{eq:criterion2}.
\end{corollary}

\begin{proof}
The specialized named predecessor simulates the source by
\cref{lem:bootstrap,thm:backend,lem:restart,prop:dispatcher}.  The final
121-state table is an all-tapes strong-bisimulation quotient of that
predecessor, so it has the same blank-tape trajectory up to state names.
Apply \cref{thm:source}.
\end{proof}

\section{Reachability-restricted reduction from 121 to 120 states}

The final state reduction is not justified by ordinary bisimulation, because
the two candidate states disagree on one read symbol. Instead, we prove that
the disagreeing state/read pair is unreachable from the blank initial
configuration. This reachability fact permits a projection from 121 to 120
states while preserving the complete blank-tape trajectory.

\subsection{Five-cell window invariant}

For a working state $q$ of the 121-state machine, let $S_q$ be a set of
five-bit words. A word records the tape cells at offsets $-2,-1,0,1,2$ from
the head, with the read cell at the center.

The initial pair is state $0$ with window $00000$. Given an admitted pair
$(q,w)$, execute the actual transition on the center symbol, shift the
five-cell window according to the head move, and admit both possibilities $0$
and $1$ for the newly exposed edge cell. A halting target has no successor
obligation. Repeating this operation reaches a finite fixed point because
there are only $121\cdot32$ possible pairs.

The delivered certificate consists of one 32-bit membership mask per old
state. It contains 1,621 admitted state-window pairs. The standalone checker
verifies the initial pair and 3,218 successor obligations. Crucially, an
unseen exterior cell is always allowed to be either symbol; no blankness
assumption is propagated beyond the observed window.

\begin{lemma}[Window invariant]
\label{lem:window}
Every state/window pair occurring at any finite time in the actual blank-tape
run of the 121-state machine belongs to the delivered certificate.
\end{lemma}

\begin{proof}
Induct on time. The blank initial configuration supplies the admitted pair
$(0,00000)$. If the actual current pair is admitted, the machine transition is
determined by the state and center bit. The next five-cell window is one of
the two successors checked by the certificate, according to the actual newly
exposed bit. Both possibilities were required to be admitted, so the actual
successor is admitted.
\end{proof}

The certificate excludes exactly two state/read pairs: $(37,0)$ and $(97,0)$.
State 97 is the canonical image of named state \texttt{6.continue.2}; state 32
is the canonical image of control state \texttt{N24}. Their read-one actions
agree. Their read-zero actions differ, but $(97,0)$ is absent from the
invariant.

\subsection{Projection theorem}

\begin{theorem}[Reachability-restricted state projection]
\label{thm:projection}
Let $M$ be a deterministic binary TM and $M'$ another such TM. Let $\pi$ map
working states of $M$ to working states of $M'$ and map halt to halt. Assume
that $\pi$ maps the initial state of $M$ to the initial state of $M'$. Suppose
a local-window invariant for $M$ contains every pair in the blank-tape run.
Assume that for every state/read pair admitted by that invariant, the
transition of $M'$ from $(\pi(q),s)$ has the same write and move as the
transition of $M$ from $(q,s)$, and its target is the $\pi$-image of the
target of $M$. Then the two blank-tape runs have identical tapes and head
positions at every time, with states related by $\pi$. They therefore halt at
exactly the same time or both run forever.
\end{theorem}

\begin{proof}
Induct on time. Initially the tapes, head positions, and projected states
agree. By \cref{lem:window}, the actually used state/read pair is in the
verified domain. The projection hypothesis therefore gives the same write and
head move and a projected target state. Hence the relation is preserved. The
halting case is included by mapping halt to halt.
\end{proof}

The delivered projection merges old states 32 and 97 and leaves the remaining
classes distinct. The standalone verifier checks the projection equations on
all 240 admitted state/read pairs. The quotient has exactly 120 states. The
arbitrary read-zero action chosen for the merged class is irrelevant to the
blank-tape run because the only old member on which it would disagree cannot
read zero under the invariant.

\begin{corollary}
\label{cor:120eq121}
The 120-state table and the 121-state table have exactly the same blank-tape
halting behavior.
\end{corollary}

\section{Proof of the main theorem}

\begin{proof}[Proof of \cref{thm:main}]
By \cref{thm:criterion}, RH holds exactly when the integer inequality
\eqref{eq:criterion2} holds for every $n\ge2$. By \cref{thm:source}, the
five-register source halts exactly on the first failure of that inequality.
By \cref{prop:tailsharing}, the revised scaling-loop grouping preserves these
source semantics. By \cref{cor:121semantics}, the delivered 121-state table has exactly this
blank-tape source semantics. By \cref{cor:120eq121}, the 120-state projection
has exactly the same blank-tape halting behavior as the 121-state table.
Therefore $M_{120}$ halts exactly when the RH-equivalent universal inequality
fails, i.e. exactly when RH is false.
\end{proof}

\section{Computational verification}

The supplied reproduction driver checks the finite construction independently
of numerical evidence for the Riemann hypothesis.
\begin{itemize}
\item The preceding 121-state machine is regenerated exactly and compared
with the archived transition table before the new specialization is applied.
\item The revised source produces 122 reachable named states, consisting of
59 framework states and 63 decision states. Its ordinary-bisimulation quotient
has 121 working states and 242 transitions.
\item The control-equivalence certificate contains 84 operation-boundary
pairs and 125 outcome edges, including a halting pair. Every possible
conditional-decrement outcome is checked.
\item The dispatcher certificate lists all 121 terminal control paths,
with maximum word length 18 and maximum consecutive-zero run three.
\item The five-cell invariant contains 1,621 admitted state-window pairs
and 3,218 closure obligations. The 121-to-120 projection is verified on all
240 admitted state/read pairs. Deliberately corrupted certificates and
reachable transition writes are rejected.
\item The final table has 120 working states and 240 transitions. Its
reversible transition integer has 2,140 bits; decoding reconstructs every
transition.
\item The literal bootstrap checker reaches the asserted initial main
configuration after 89,775,610 transitions. A literal tape runner then checks
completed endpoints $2,3,4$ at 92,233,600, 113,387,256, and 208,951,810
transitions, respectively, with the exact expected live-register tuples.
\item Exact host-integer checks cover 255 complete stages through $n=256$,
32,640 intermediate product identities, 100,000 scaling-clock inputs, and
10,000 unit-conserving scaling cases. These use the proved counter contracts
and are distinguished from literal tape executions.
\end{itemize}

These computations verify finite implementation claims. The infinite-run
statement follows from the analytic criterion, the uniform register and
control semantics, and the inductive state-projection theorem.

\section{Separate NQL implementation}
\label{sec:nqlcheck}

For the preceding 121-state version, a cross-check using a distinct implementation path was performed: the same counterexample search was
written in ordinary NQL and compiled with an unmodified snapshot of the public
NQL compiler \cite{NQL}. This verification implementation intentionally omits
the program-specific dispatcher, specialized register microcode, and
reachability-restricted projection used to obtain the 120-state count.

The NQL program computes the same quantities $A_n-1$, $\beta_n(A_n-1)$,
$\Delta_n$, and the same simultaneous-scaling test of \cref{eq:criterion}. The
unmodified compiler produces 1,102 raw reachable states and 601 states after
its own compression, with 1,202 transitions. A further ordinary
strong-bisimulation pass does not reduce the 601-state count. Thus the
separate implementation is intentionally much larger:
\[
601\text{ states (plain NQL)}\qquad\text{versus}\qquad120\text{ states (specialized construction)}.
\]
The 601-state implementation is not intended as a state-count competitor. Its
purpose is to provide a separate implementation of the same mathematical
criterion using a conventional compiler, thereby distinguishing correctness of
the arithmetic search from correctness of the specialized 120-state
implementation.

The NQL source was parsed with the actual NQL grammar and executed with a
small generic natural-number AST interpreter. Against an independent oracle
computing $A_n=n!P(n)$ from a sieve rather than the self-sieve, it passed all
stages $n=2,\ldots,256$, comprising 255 complete stage checks and 32,640
intermediate product identities. Separately, assertion-heavy NQL helper
programs were compiled to binary blank-tape Turing machines and executed
successfully for $n=2,3,4,5$; the $n=4$ case covers the exceptional composite
and $n=5$ covers the following prime update.

This is a separate implementation check of the mathematical search, not a
proof that the 120-state table and the 601-state table are step-for-step
equivalent. The specialized 120-state path is instead justified by the source
invariants, dispatcher certificate, and reachability projection proved in this
paper. The two verification routes are complementary.

\section{Reproducibility}

The complete source required to regenerate and verify the principal finite
artifacts is printed in the appendices and included in the accompanying
directory. The principal shell commands are
\begin{lstlisting}[language=bash]
python reproduce.py --test
python verify_certificate.py
g++ -O3 -std=c++17 bootstrap_check.cpp -o bootstrap_check
./bootstrap_check machines/RH_121_raw_named.tm
g++ -O3 -std=c++17 stage_check.cpp -o stage_check
./stage_check machines/RH_121_raw_named.tm
\end{lstlisting}
The exact transition table in \cref{app:table} can be inspected independently
of the generator. The accompanying directory contains the revised generator, transition tables,
source-control certificate, projection certificates, and literal tape checkers.
The separate NQL results in \cref{sec:nqlcheck} concern the unchanged arithmetic
criterion and are reported from the preceding verification work.

\section{Discussion and limitations}

The present construction reduces the public 744-state RH benchmark to an
explicit 120-state machine in the same binary blank-tape model. The reduction
is obtained by combining an RH-equivalent arithmetic reformulation with a
compact five-register implementation and two program-specific tape reductions.
In particular, $A_n$ combines factorial growth with prime information, the
self-sieving recurrence eliminates a separate primality routine, the offset
digit count replaces explicit logarithmic evaluation, and the scaling
procedure implements the required error allowance by repeated halving. The
program-specific dispatcher and the reachability-restricted projection then
reduce the tape-level control overhead.

The separate 601-state NQL implementation provides a useful comparison: the
same mathematical search can be executed through a conventional compiler
without the specialized state reductions. The substantial difference between
601 and 120 states therefore reflects the effect of representation and
implementation choices rather than a different stopping criterion.

No lower bound is proved, and the construction does not imply that 120 states
is close to optimal. Nor does it prove RH. The theorem establishes an explicit
equivalence: the 120-state machine halts exactly when the RH-equivalent
universal integer criterion in \cref{thm:criterion} fails. The public sources
checked on 22 September 2026 identify 744 states as the previously documented
benchmark; accordingly, 120 is a candidate new state-count record and, subject
to independent verification, the smallest publicly documented explicit
construction found in the sources examined here.

The analytic forward implication uses Schoenfeld's RH-conditional estimate in
\cref{thm:Schoenfeld}; the converse uses analytic continuation and the
functional equation of the Riemann zeta function. The finite source-to-tape
machinery is included in full but is not formalized in a proof assistant.
Independent audit of the compiler path, dispatcher specialization, and
reachability certificate is therefore appropriate before treating the
state-count claim as certified.

\appendix
\clearpage
\section{Final high-level source}
\label{app:source}

The appendices provide the complete implementation artifacts underlying the
finite construction. The final source consists of the baseline five-register
program, a subclass that modifies the deficit-seeding phase, and a final
subclass that shares the scaling-loop continuation. All three are reproduced verbatim to permit line-by-line comparison with the preceding
proofs.

\subsection{Baseline source}

The following Python source is \texttt{baseline\_source.py}.
\begin{lstlisting}[language=Python,basicstyle=\ttfamily\scriptsize,breaklines=true,columns=fullflexible,keepspaces=true]
"""Five-register RH counterexample search with an offset-division digit count.

No uncounted arithmetic operations, tape input or infinite background are used.
At main entry m=M and all other registers are zero. A passing stage tests n=M+2
and returns with m=n-1, all other registers zero. The inherited root restarts main.
"""
from pathlib import Path
import sys
sys.path.insert(0, str(Path(__file__).resolve().parent/'vendor'))
from builder import Builder, subroutine
from dispatch import specialize_dispatch

class RHBuilder(Builder):
    def __init__(self, fast=True):
        super().__init__()
        self.fast = fast
        self.dispatch_certificate = None
        for name in ('m','j','x','f','v'):
            setattr(self,name,self.reg(name))

    def label(self,name,seq,zeros=0,ones=0):
        while isinstance(seq,(tuple,list)) and len(seq)==1:
            seq=seq[0]
        if isinstance(seq,int):
            if name not in self.seqs[seq].unresolved:
                return seq
            seq=self.seqs[seq].seq
        return super().label(name,seq,zeros,ones)

    @subroutine
    def move(self,a,b):
        """b += a, a = 0, for distinct registers."""
        return [self.while_decnz(a,b.inc)]

    @subroutine
    def addval(self,a,b,tmp):
        """b += a, preserving a; distinct registers; tmp starts/ends zero."""
        return [self.while_decnz(a,[b.inc,tmp.inc]),self.move(tmp,a)]

    @subroutine
    def fold(self):
        """(x,j,v) <- (x+j,j+v,0). The two increments commute."""
        loop=self.label('fold',[self.j.decnz,
            [self.v.inc,[self.x.inc,'continue_fold']]])
        return [loop,self.move(self.v,self.j)]

    @subroutine
    def transform(self):
        """(x,j,v) <- (x*(j+v+1)+j,j+v,0); m starts/ends zero.

        Fold first, then repeat 'x += 1; fold' once for each saved input unit.
        The post-tested loop is algebraically the earlier prelude-plus-loop.
        """
        loop=self.label('transform',[self.fold(),
            [self.m.decnz,[self.x.inc,'continue_transform']]])
        return [self.move(self.x,self.m),loop]

    @subroutine
    def twice(self):
        return [self.transform(),self.transform()]

    @subroutine
    def divide(self):
        """For m=0, divide B=x by J=j+v+1, even when v is initially nonzero.

        Output (x,j,v)=(B//J,B%J,J-1-B%J); m is again zero.
        """
        tick=self.label('tick',[[self.v.decnz,[self.j.inc,'break_tick']],
            [self.move(self.j,self.v),self.x.inc]])
        return [self.move(self.j,self.v),self.move(self.x,self.m),
                self.while_decnz(self.m,tick)]

    @subroutine
    def if_nonzero(self,a,body):
        return [self.if_decnz(a,[a.inc,body])]

    @subroutine
    def product_step(self):
        """Advance J and compute A_J-1, where A_J=J! product_{p<=J}p."""
        return [self.j.inc,self.divide(),
                self.if_nonzero(self.v,self.transform()),self.twice()]

    @subroutine
    def dec_f(self):
        """Saturating decrement, returning normally on either result."""
        return [self.label('sat',[self.f.decnz,'break_sat'])]

    @subroutine
    def digits(self):
        """Subtract b_j+1(x) from f, saturating, where b is defined by
        x <- (x-1)//(j+1) while x>0. Consume x and restore the divisor.

        Unlike ordinary digit counting, the loop condition's decrement is
        intentional. Remainder and complement remain split between divisions.
        """
        loop=self.label('digits',[self.x.decnz,
            [self.divide(),[self.dec_f(),'continue_digits']]])
        return [loop,self.move(self.v,self.j)]

    @subroutine
    def halve(self,a):
        """a=j or f: a <- floor(a/2), v=0 scratch. For j, collect removed
        units in m. Register x is untouched and can count scaling iterations.
        """
        prefix=[self.m.inc] if a==self.j else []
        return [self.while_decnz(a,[prefix,self.if_decnz(a,self.v.inc)]),
                self.move(self.v,a)]

    @subroutine
    def compare(self):
        return [self.while_decnz(self.x,self.dec_f()),self.if_decnz(self.f,'halt')]

    @subroutine
    def width(self):
        """At entry j=n-1, f=delta, other registers zero. Test
        delta//2**q <= q, q=floor(log_4(3n)), and recover n-1 in m.
        """
        return [self.while_decnz(self.j,[self.m.inc,self.halve(self.j),
            self.halve(self.j),self.halve(self.f),self.x.inc]),self.compare()]

    def stage_code(self):
        prod=self.label('product',[self.product_step(),[self.f.decnz,'continue_product']])
        return [self.move(self.m,self.f),prod,
                self.addval(self.j,self.f,self.v),self.digits(),self.width()]

    @subroutine
    def main(self):
        return self.stage_code()

    def generate(self):
        table=super().generate()
        if self.fast:
            table,self.dispatch_certificate=specialize_dispatch(self,table)
        return table
\end{lstlisting}

\subsection{Split-divisor source specialization}

The following Python source is \texttt{rh\_source.py}.
\begin{lstlisting}[language=Python,basicstyle=\ttfamily\scriptsize,breaklines=true,columns=fullflexible,keepspaces=true]
"""RH source with an unrecombined split divisor at the digit-phase entrance.

The mathematical criterion and order of endpoints are unchanged from RH_123.
The inherited source/backend and dispatcher are in baseline_source.py and vendor/.
"""
from baseline_source import RHBuilder as Baseline, subroutine

class RHBuilder(Baseline):
    @subroutine
    def seed_deficit(self):
        """For f=v=0: (j,f,v)=(J,0,0) -> (0,J,J).

        The divisor j+v+1 is unchanged. The next divide accepts this split form.
        """
        return [self.while_decnz(self.j,[self.v.inc,self.f.inc])]

    def stage_code(self):
        prod=self.label('product',[self.product_step(),[self.f.decnz,'continue_product']])
        return [self.move(self.m,self.f),prod,self.seed_deficit(),
                self.digits(),self.width()]
\end{lstlisting}

\subsection{Scaling-loop specialization}

The following final subclass changes only the binary grouping of the scaling
loop. The complete generated control graph has one fewer decision state.
The preceding split-divisor subclass is the Python file
\texttt{rh\_source.py}; this final subclass is the Python file
\texttt{rh120\_source.py}.

\begin{lstlisting}[language=Python,basicstyle=\ttfamily\scriptsize]
"""120-state RH source: share the scaling loop's increment/continue tail.

All counter operations and the tested RH criterion are unchanged from RH_121.
Only the binary grouping is changed before the local jump is resolved.
"""
from rh_source import RHBuilder as Previous, subroutine

class RHBuilder(Previous):
    @subroutine
    def width(self):
        prefix=[self.m.inc,self.halve(self.j),self.halve(self.j),self.halve(self.f)]
        loop=self.label('width',[self.j.decnz,[prefix,[self.x.inc,'continue_width']]])
        return [loop,self.compare()]
\end{lstlisting}

\section{Program-specific dispatcher}
\label{app:dispatch}

The following Python source is \texttt{dispatch.py}.
\begin{lstlisting}[language=Python,basicstyle=\ttfamily\scriptsize,breaklines=true,columns=fullflexible,keepspaces=true]
"""A checked, program-specific replacement of the read-only dispatcher scan.

This is not an all-tapes bisimulation: its premise is the live control-word
language of the compiled source. See README for the invariant and bootstrap check.
"""
def control_words(builder):
    rows=[]
    def visit(node,prefix,active):
        if isinstance(node,int):
            if node in active:
                raise ValueError('This certificate requires an acyclic control DAG')
            seq=builder.seqs[node].seq
            if len(seq)!=2:raise ValueError('Source has not been binarized')
            for bit,child in enumerate(seq):
                visit(child,prefix+str(bit),active|{node})
        else:
            rows.append((prefix,node))
    visit(builder.root,'110',frozenset())
    return rows

def specialize_dispatch(builder,table):
    words=control_words(builder)
    longest=max(max(map(len,w.split('1'))) for w,_ in words)
    # Boot2 is alternating control with a two-one suffix. Normal operation
    # adds a phase marker 1 and at most one intervening zero before the sentinel.
    k=max(longest,2)+1
    out=dict(table)
    for i in range(k):
        out[f'4.run.{i}',1]=(1,'L','4.run.0')
        out[f'4.run.{i}',0]=(0,'L',f'4.run.{i+1}') if i+1<k else (0,'R','4.dispatch.scan')
    old_entry=out['6.continue.0',1]
    out['6.continue.0',1]=(1,'L','4.run.0')
    reached=set();todo=['0a.boot1.A']
    while todo:
        q=todo.pop()
        if q in reached:continue
        reached.add(q)
        for bit in (0,1):
            if (q,bit) in out:todo.append(out[q,bit][2])
    out={key:v for key,v in out.items() if key[0] in reached}
    old_states={q for q,_ in table};new_states={q for q,_ in out}
    cert={'live_control_prefix':'110','terminal_paths':len(words),
        'maximum_word_length':max(len(w) for w,_ in words),
        'maximum_consecutive_zeros':longest,'scanner_states':k,
        'old_entry_transition':old_entry,'new_entry_transition':out['6.continue.0',1],
        'removed_states':sorted(old_states-new_states),'added_states':sorted(new_states-old_states),
        'control_words':[{'bits':w,'terminal':v} for w,v in words],
        'scope':'Language/invariant-based scanner replacement, not an all-tapes state merge.'}
    return out,cert
\end{lstlisting}

\section{Register-machine backend}
\label{app:builder}

The following Python source, \texttt{builder.py}, is the complete
operational builder used for the selected source. It is a research
adaptation of Andrew J. Wade's MIT-0 TMBuilder; the provenance notice is
retained in the file.
\begin{lstlisting}[language=Python,basicstyle=\ttfamily\tiny,breaklines=true,columns=fullflexible,keepspaces=true]
"""Research adaptation of Andrew J. Wade's MIT-0 TMBuilder.
Source: https://raw.githubusercontent.com/LegionMammal978/turing_machine_explorer/main/TMBuilder.py
The naming/debugger layer is replaced; operational framework and IR follow Wade.
"""
from functools import lru_cache
from dataclasses import dataclass

class Reg(str):
    @property
    def inc(self): return '3.'+self+'.inc'
    @property
    def decnz(self): return '3.'+self+'.decnz'

@dataclass
class Seq:
    seq: tuple
    name: str | None = None
    unresolved: frozenset = frozenset()

def subroutine(f):
    def g(self,*args):
        args=tuple(self.add(v) for v in args)
        key=(f.__name__,*args)
        if key in self.memo: return self.memo[key]
        rv=tuple(v for v in (self.add(v2) for v2 in f(self,*args)) if v!=())
        self.memo[key]=rv
        if rv:
            k=self.add(rv)
            if isinstance(k,int): self.seqs[k].name=f.__name__
        return rv
    return g

class Builder:
    def __init__(self):
        self.regs={};self.seqs=[];self.lookup={};self.memo={};self.root=None
    def reg(self,name):
        if name not in self.regs: self.regs[name]=Reg(name)
        return self.regs[name]
    def add(self,tree):
        while isinstance(tree,(tuple,list)) and len(tree)==1: tree=tree[0]
        if not isinstance(tree,(tuple,list)): return tree
        if not tree: return ()
        key=tuple(v for v in (self.add(x) for x in tree) if v!=())
        if not key: return ()
        if len(key)==1: return key[0]
        if key not in self.lookup:
            unresolved=set()
            for v in key:
                if isinstance(v,int): unresolved.update(self.seqs[v].unresolved)
                elif v.startswith(('break_','continue_')): unresolved.add(v.split('_',1)[1])
            self.lookup[key]=len(self.seqs)
            self.seqs.append(Seq(key,None,frozenset(unresolved)))
        return self.lookup[key]
    def label(self,name,seq,zeros=0,ones=0):
        cf=('break_'+name,'continue_'+name)
        def contains(s):
            for e in s:
                if e in cf: return True
                if isinstance(e,int) and name in self.seqs[e].unresolved: return True
                if isinstance(e,(tuple,list)) and contains(e): return True
            return False
        def resolve(z,o,v):
            if v in cf: return '_break.'+str(z) if v==cf[0] else '_continue.'+str(o)
            if isinstance(v,int) and name in self.seqs[v].unresolved:
                return self.add(self.label(name,self.seqs[v].seq,z,o))
            if isinstance(v,(tuple,list)): return self.label(name,v,z,o)
            return v
        if not contains(seq): return seq
        assert len(seq)==2,(name,seq)
        return (resolve(zeros+1,ones,seq[0]),resolve(zeros,ones+1,seq[1]))
    def reachable(self):
        grey={self.root};black=set()
        while grey:
            i=grey.pop()
            if i not in black:
                yield i;black.add(i)
                grey.update(v for v in self.seqs[i].seq if isinstance(v,int))
    def breakout(self):
        length=max(len(s.seq) for s in self.seqs)
        while length>1:
            seen=set();match=None;reach=list(self.reachable())
            for i in reach:
                s=self.seqs[i].seq
                for off in range(len(s)+1-length):
                    p=s[off:off+length]
                    if p in seen: match=p;break
                    seen.add(p)
                if match: break
            if match:
                mi=self.add(match)
                for i in reach:
                    while True:
                        s=self.seqs[i].seq
                        if len(s)<=length: break
                        for off in range(len(s)+1-length):
                            if s[off:off+length]==match:
                                self.seqs[i].seq=s[:off]+(mi,)+s[off+length:];break
                        else: break
            else: length-=1
    def binary(self):
        i=0
        while i<len(self.seqs):
            s=self.seqs[i].seq
            if len(s)>2: self.seqs[i].seq=(s[0],self.add(s[1:]))
            i+=1
    @subroutine
    def while_decnz(self,var,body):
        return [self.label('loop',[var.decnz,'continue_loop' if body==() else [body,'continue_loop']])]
    @subroutine
    def if_decnz(self,var,body): return [[var.decnz,body]]
    @subroutine
    def if_not_decnz(self,var,body): return [self.label('fn',[[var.decnz,'break_fn'],body])]
    @subroutine
    def pair(self,out,a,b):
        return [self.label('loop',[self.while_decnz(a,[b.inc,out.inc]),[b.decnz,[[out.inc,self.while_decnz(b,a.inc)],'continue_loop']]])]
    @subroutine
    def unpair(self,a,b,n):
        return [self.label('loop',[n.decnz,[a.inc,[[b.decnz,'continue_loop'],[self.while_decnz(a,b.inc),'continue_loop']]]])]
    @subroutine
    def if_eq(self,a,b,body):
        return [self.label('fn',[[self.label('loop',[a.decnz,[[b.decnz,'continue_loop'],[self.while_decnz(a,()),'break_fn']]]),[b.decnz,[self.while_decnz(b,()),'break_fn']]],body])]
    def debug_assert_eq_val(self,*args): return ()
    def build(self,optimize=True):
        self.root=self.add(self.main())
        if optimize: self.breakout()
        self.binary()
        return self.generate()
    def generate(self):
        table={}
        def emit(q,s,w,d,n): table[q,int(s)]=(int(w),d,n)
        # Compact transcription of the operational framework, same labels as source.
        rows='''
0a.boot1.A 0 1 R 0a.boot1.B
0a.boot1.A 1 1 L 0a.boot1.C
0a.boot1.B 0 0 L 0a.boot1.A
0a.boot1.B 1 0 L 0a.boot1.D
0a.boot1.C 0 1 L 0a.boot1.A
0a.boot1.D 0 1 L 0a.boot1.B
0a.boot1.D 1 1 R 0a.boot1.E
0a.boot1.E 0 0 R 0a.boot1.D
0a.boot1.E 1 0 R 0a.boot1.B
0b.boot2.0 0 1 L 6.continue.0
0b.boot2.0 1 1 R 0b.boot2.1
0b.boot2.1 0 0 R 0b.boot2.0
0b.boot2.1 1 1 R 0b.boot2.2
1b.reg.prep_1 0 0 R 1b.reg.prep_2
1b.reg.prep_2 0 1 R 1b.reg.prep_2
1b.reg.prep_2 1 1 L 6.continue.0
2b.reg.-1.dec 0 0 R 2b.reg.-2.dec
2b.reg.-1.dec 1 1 R 2b.reg.-1.dec
2b.reg.-1.inc 0 0 R 2b.reg.-2.inc
2b.reg.-1.inc 1 1 R 2b.reg.-1.inc
2b.reg.-2.dec 0 1 L 2c.reg.return_1_1
2b.reg.-2.dec 1 0 R 2b.reg.dec.check
2b.reg.-2.inc 0 1 R 2b.reg.inc.shift_1
2b.reg.-2.inc 1 1 R 2b.reg.-2.inc
2b.reg.dec.check 0 0 L 2b.reg.-2.dec
2b.reg.dec.check 1 1 R 2b.reg.dec.scan_1
2b.reg.dec.scan_1 0 0 R 2b.reg.dec.scan_2
2b.reg.dec.scan_1 1 1 R 2b.reg.dec.scan_1
2b.reg.dec.scan_2 0 0 L 2b.reg.dec.shift_1
2b.reg.dec.scan_2 1 1 R 2b.reg.dec.scan_1
2b.reg.dec.shift_1 0 1 L 2b.reg.dec.shift_2
2b.reg.dec.shift_1 1 1 L 2b.reg.dec.shift_1
2b.reg.dec.shift_2 0 0 L 2c.reg.return_0_1
2b.reg.dec.shift_2 1 0 L 2b.reg.dec.shift_1
2b.reg.inc.shift_1 0 0 L 2c.reg.return_0_1
2b.reg.inc.shift_1 1 0 R 2b.reg.inc.shift_2
2b.reg.inc.shift_2 0 1 R 2b.reg.inc.shift_1
2b.reg.inc.shift_2 1 1 R 2b.reg.inc.shift_2
2c.reg.return_0_1 0 0 L 2c.reg.return_0_2
2c.reg.return_0_1 1 1 L 2c.reg.return_0_1
2c.reg.return_0_2 0 0 L 6.break.0
2c.reg.return_0_2 1 1 L 2c.reg.return_0_1
2c.reg.return_1_1 0 0 L 2c.reg.return_1_2
2c.reg.return_1_1 1 1 L 2c.reg.return_1_1
2c.reg.return_1_2 0 0 L 6.break.1
2c.reg.return_1_2 1 1 L 2c.reg.return_1_1
2e.reg.cleanup_1 0 0 R 2e.reg.cleanup_1
2e.reg.cleanup_1 1 0 R 2e.reg.cleanup_2
2e.reg.cleanup_2 0 0 L 0b.boot2.0
2e.reg.cleanup_2 1 0 R 2e.reg.cleanup_2
4.dispatch.0 0 0 L 4.dispatch.scan
4.dispatch.0 1 1 L 4.dispatch.0
4.dispatch.scan 0 0 R 4.dispatch.scan
4.dispatch.scan 1 1 R 5.root.1
5.root.1 0 0 R 0b.boot2.0
5.root.1 1 1 R 5.root.2
6.break.0 0 1 R 2e.reg.cleanup_1
6.break.0 1 0 L 6.break.0
6.break.1 0 0 L 6.break.0
6.break.1 1 0 L 6.break.1
6.continue.0 0 0 L 6.continue.0
'''
        for line in rows.strip().splitlines(): emit(*line.split())
        for j,r in enumerate(self.regs.values()):
            for op in ('inc','dec'):
                emit(f'2b.reg.{j}.{op}',0,0,'R',f'2b.reg.{j-1}.{op}')
                emit(f'2b.reg.{j}.{op}',1,1,'R',f'2b.reg.{j}.{op}')
            for state,op in ((r.inc,'inc'),(r.decnz,'dec')):
                emit(state,0,1,'R','1b.reg.prep_1')
                emit(state,1,0,'R',f'2b.reg.{j}.{op}')
        r=next(iter(self.regs.values()))
        emit('0a.boot1.C',1,1,'R',r.inc)
        emit('0b.boot2.2',0,0,'R',r.inc)
        emit('0b.boot2.2',1,1,'R',r.decnz)
        @lru_cache(None)
        def zeros(v):
            if not isinstance(v,int): return 0
            a,b=self.seqs[v].seq
            return max(1+zeros(a),zeros(b))
        mz=max(2+zeros(self.root),2)
        for j in range(1,mz+1):
            emit(f'4.dispatch.{j}',0,0,'L',f'4.dispatch.{j-1}')
            emit(f'4.dispatch.{j}',1,1,'L',f'4.dispatch.{j}')
        emit('6.continue.0',1,1,'L',f'4.dispatch.{mz}')
        for s in (0,1): emit('5.root.2',s,0,'R',f'N{self.root}')
        def jump(kind,j):
            for k in range(1,j+1):
                if kind=='break':
                    emit(f'6.break.{k}',0,0,'L',f'6.break.{k-1}')
                    emit(f'6.break.{k}',1,0,'L',f'6.break.{k}')
                else:
                    emit(f'6.continue.{k}',0,0,'L',f'6.continue.{k}')
                    emit(f'6.continue.{k}',1,0,'L',f'6.continue.{k-1}')
        for i in self.reachable():
            for s,v in enumerate(self.seqs[i].seq):
                q=f'N{i}'
                if isinstance(v,int): emit(q,s,s,'R',f'N{v}')
                elif v.startswith('_break.'):
                    level=int(v.split('.')[1])
                    if s==0 and level==0: emit(q,s,1,'L','6.continue.0')
                    else:
                        j=level-(1-s); assert j>=0
                        jump('break',j);emit(q,s,0,'L',f'6.break.{j}')
                elif v.startswith('_continue.'):
                    j=int(v.split('.')[1])-s;assert j>=0
                    jump('continue',j);emit(q,s,0,'L',f'6.continue.{j}')
                else: emit(q,s,s,'R',v)
        reach=set();stack=['0a.boot1.A']
        while stack:
            q=stack.pop()
            if q in reach: continue
            reach.add(q)
            stack += [table[q,s][2] for s in (0,1) if (q,s) in table]
        return {k:v for k,v in table.items() if k[0] in reach}

def count(table):
    states={q for q,s in table}
    return {'framework':sum(q[0].isdigit() for q in states),'decision_DAG':sum(q[0]=='N' for q in states),'total':len(states)}

def save(table,path):
    with open(path,'w') as f:
        f.write('#! start 0a.boot1.A\n')
        for (q,s),(w,d,n) in sorted(table.items()): f.write(f'{q} {s} {w} {d} {n}\n')
\end{lstlisting}

\section{All-tapes minimizer and table utilities}

\subsection{Strong-bisimulation minimizer}

The following Python source is \texttt{minimize.py}.
\begin{lstlisting}[language=Python,basicstyle=\ttfamily\scriptsize,breaklines=true,columns=fullflexible,keepspaces=true]
"""Strong bisimulation minimization: preserves every tape trajectory, up to states.
No transition deletion, heuristic halting assumptions, or finite-run equivalence.
"""
def minimize(t,start='0a.boot1.A'):
    states=sorted({q for q,s in t})
    # All names without outgoing instructions are halts.
    parts={q:0 for q in states}
    rounds=0
    while True:
        keys={};nparts={}
        for q in states:
            key=tuple(None if (q,s) not in t else (t[q,s][0],t[q,s][1],parts.get(t[q,s][2],-1)) for s in (0,1))
            nparts[q]=keys.setdefault(key,len(keys))
        rounds+=1
        if nparts==parts:break
        parts=nparts
    reps={}
    for q in states:reps.setdefault(parts[q],q)
    reps[parts[start]]=start
    mapping={q:reps[parts[q]] for q in states}
    out={}
    for (q,s),(w,d,n) in t.items():
        k=mapping[q],s;value=(w,d,mapping.get(n,'halt'))
        assert k not in out or out[k]==value
        out[k]=value
    return out,mapping,rounds
\end{lstlisting}
\subsection{Canonical table utilities}

The following Python source is \texttt{tables.py}.
\begin{lstlisting}[language=Python,basicstyle=\ttfamily\scriptsize,breaklines=true,columns=fullflexible,keepspaces=true]
from collections import deque
from pathlib import Path
import hashlib
def canonical(table,start='0a.boot1.A'):
    states={q for q,s in table}
    assert {n for w,d,n in table.values()}-states <= {'halt'}
    labels={start:0};queue=deque([start]);out={};missing=[]
    while queue:
        q=queue.popleft()
        for s in (0,1):
            if (q,s) not in table:
                assert (q,s)==('1b.reg.prep_1',1)
                w,d,n=1,'R','halt';missing.append((q,s))
            else:w,d,n=table[q,s]
            assert w in (0,1) and d in ('L','R')
            if n=='halt':dest=-1
            else:
                if n not in labels:labels[n]=len(labels);queue.append(n)
                dest=labels[n]
            out[labels[q],s]=(w,1 if d=='R' else -1,dest)
    assert len(labels)==len(states)
    return out,labels,missing

def write(table,path):
    text='# start 0; blank 0; halt H; one tape; two symbols; L/R only\n'
    for (q,s),(w,d,n) in sorted(table.items()):
        text+=f'{q} {s} {w} {"R" if d==1 else "L"} {"H" if n==-1 else n}\n'
    path.write_text(text)
    return hashlib.sha256(text.encode()).hexdigest()


def read(path):
    t={}
    for line in Path(path).read_text().splitlines():
        if not line or line.startswith('#'):continue
        q,s,w,d,n=line.split();key=(int(q),int(s))
        assert key not in t
        t[key]=(int(w),1 if d=='R' else -1,-1 if n=='H' else int(n))
    n=len(t)//2
    assert len(t)==2*n and set(t)=={(q,s) for q in range(n) for s in (0,1)}
    assert all(w in (0,1) and d in (-1,1) and -1<=z<n for w,d,z in t.values())
    return t
\end{lstlisting}

\section{Window certificate and reachability-restricted quotient}

\subsection{Window invariant generator and verifier}

The following Python source is \texttt{window\_certificate.py}.
\begin{lstlisting}[language=Python,basicstyle=\ttfamily\scriptsize,breaklines=true,columns=fullflexible,keepspaces=true]
"""Proof-producing local-window reachability for deterministic binary TMs.

A window contains 2*r+1 consecutive cells, with its most-significant bit leftmost.
The unknown newly exposed cell is ALWAYS allowed to be either symbol. This is an
inductive overapproximation, not a finite execution sample or a nonhalting test.
"""
from collections import deque

def successors(table,q,w,r):
    symbol=(w>>r)&1
    write,move,target=table[q,symbol]
    if target<0:return ()
    changed=(w&~(1<<r))|(write<<r)
    mask=(1<<(2*r+1))-1
    if move==1:return ((target,((changed<<1)&mask)|b) for b in (0,1))
    return ((target,(changed>>1)|(b<<(2*r))) for b in (0,1))

def generate(table,radius=2):
    if not isinstance(radius,int) or not 1<=radius<=8:raise ValueError('radius')
    n=len(table)//2;mask=[0]*n;mask[0]=1;todo=deque([(0,0)])
    while todo:
        q,w=todo.popleft()
        for a,v in successors(table,q,w,radius):
            if not(mask[a]>>v)&1:
                mask[a]|=1<<v;todo.append((a,v))
    return {'radius':radius,'window_width':2*radius+1,'masks':mask,
        'interpretation':'bit w of masks[q] permits the state/window pair (q,w); left cell is highest bit',
        'initial_state':0,'initial_tape':'all zero',
        'unknown_exterior':'both symbols, independently at every abstract step'}

def verify(table,cert):
    """Standalone finite invariant check, independent of the fixed-point finder."""
    r=cert['radius'];W=2*r+1;n=len(table)//2;m=cert['masks']
    if not isinstance(r,int) or not 1<=r<=8 or cert['window_width']!=W:raise ValueError('bad radius')
    if len(m)!=n or any(not isinstance(x,int) or x<0 or x>=(1<<(1<<W)) for x in m):raise ValueError('bad masks')
    if not m[0]&1:raise ValueError('blank initial configuration omitted')
    total=0;edges=set();obligations=0
    for q,bits in enumerate(m):
        for w in range(1<<W):
            if not (bits>>w)&1:continue
            total+=1;s=(w>>r)&1;edges.add((q,s))
            write,move,target=table[q,s]
            if target<0:continue
            # Spell out the successor calculation rather than call successors().
            cells=[(w>>i)&1 for i in reversed(range(W))]
            cells[r]=write
            for exterior in (0,1):
                shifted=cells[1:]+[exterior] if move==1 else [exterior]+cells[:-1]
                v=0
                for cell in shifted:v=2*v+cell
                if not (m[target]>>v)&1:raise ValueError(('invariant not closed',q,w,target,v))
                obligations+=1
    return {'state_window_pairs':total,'reachable_read_overapproximation':sorted(edges),
            'closure_obligations':obligations,
            'excluded_reads':sorted(set(table)-edges)}

def verify_projection(old,new,cert,mapping):
    evidence=verify(old,cert)
    if len(mapping)!=len(old)//2 or mapping[0]!=0:raise ValueError('bad mapping')
    if any(not isinstance(q,int) or not 0<=q<len(new)//2 for q in mapping):raise ValueError('bad state target')
    for q,s in evidence['reachable_read_overapproximation']:
        w,d,n=old[q,s]
        expected=(w,d,-1 if n<0 else mapping[n])
        if new[mapping[q],s]!=expected:raise ValueError(('projection action mismatch',q,s))
    return evidence
\end{lstlisting}
\subsection{Compatible merge search and quotient}

The following Python source is \texttt{compatible\_merge.py}.
\begin{lstlisting}[language=Python,basicstyle=\ttfamily\scriptsize,breaklines=true,columns=fullflexible,keepspaces=true]
"""Greedy compatible state merging using a supplied proven reachable-read set.
Selection is heuristic. Correctness of the selected mapping is checked separately.
"""
from collections import deque
def unionmerge(tab,edges,groups,pairs):
    n=len(tab)//2;p=list(range(n));dd=[{} for _ in range(n)]
    for q,s in edges:dd[q][s]=tab[q,s]
    def find(q):
        while p[q]!=q:p[q]=p[p[q]];q=p[q]
        return q
    todo=list(pairs)+[(q,r) for q,r in enumerate(groups) if q!=r]
    while todo:
        a,b=todo.pop()
        if a<0 or b<0:
            if a!=b:return None
            continue
        a,b=find(a),find(b)
        if a==b:continue
        if a>b:a,b=b,a
        for s,t in dd[b].items():
            if s in dd[a]:
                z=dd[a][s]
                if z[:2]!=t[:2]:return None
                todo.append((z[2],t[2]))
            else:dd[a][s]=t
        p[b]=a;dd[b]={}
    return [find(q) for q in range(n)]

def optimize(tab,edges):
    n=len(tab)//2;g=list(range(n));record=[]
    while True:
        roots=sorted(set(g));best=None
        for i,a in enumerate(roots):
            for b in roots[i+1:]:
                new=unionmerge(tab,edges,g,[(a,b)])
                if new is not None and (best is None or len(set(new))<len(set(best[0]))):best=(new,a,b)
        if best is None:break
        g,a,b=best
        record.append({'a':a,'b':b,'states':len(set(g))})
        print(record[-1],flush=True)
    return g,record

def quot(tab,edges,g):
    small={}
    for q,s in edges:
        w,d,n=tab[q,s];a=(w,d,g[n] if n>=0 else -1)
        key=(g[q],s)
        assert key not in small or small[key]==a
        small[key]=a
    # Arbitrary unreachable reads are completed, never used on actual runs.
    for old in range(len(g)):
        for s in (0,1):
            w,d,n=tab[old,s]
            small.setdefault((g[old],s),(w,d,g[n] if n>=0 else -1))
    start=g[0];ren={start:0};todo=deque([start]);out={}
    while todo:
        q=todo.popleft()
        for s in (0,1):
            w,d,n=small[q,s]
            if n>=0 and n not in ren:ren[n]=len(ren);todo.append(n)
            out[ren[q],s]=(w,d,ren[n] if n>=0 else -1)
    assert len(ren)==len(set(g))
    return out,[ren[r] for r in g]

\end{lstlisting}
\subsection{Standalone certificate checker}

The following Python source is \texttt{verify\_certificate.py}.
\begin{lstlisting}[language=Python,basicstyle=\ttfamily\scriptsize,breaklines=true,columns=fullflexible,keepspaces=true]
"""Check only delivered tables/certificates, without importing the machine compiler."""
from pathlib import Path
from copy import deepcopy
import json,hashlib
from tables import read
from window_certificate import verify_projection
ROOT=Path(__file__).resolve().parent

def run():
    old=read(ROOT/'machines/RH_121_predecessor.tm');new=read(ROOT/'machines/RH_120.tm')
    cert=json.loads((ROOT/'machines/window_certificate.json').read_text())
    mp=json.loads((ROOT/'machines/RH_120.projection.json').read_text())
    info=verify_projection(old,new,cert,mp);rejected=[]
    def must_fail(label,changed_new,changed_cert,changed_map):
        try:verify_projection(old,changed_new,changed_cert,changed_map)
        except (ValueError,KeyError,IndexError):rejected.append(label);return
        raise RuntimeError('Invalid certificate accepted: '+label)
    c=deepcopy(cert);c['masks'][0]&=~1
    must_fail('remove blank initial window',new,c,mp)
    q=next(i for i,m in enumerate(cert['masks']) if i!=0 and m)
    c=deepcopy(cert);c['masks'][q]=0
    must_fail('remove reachable successor state windows',new,c,mp)
    bad=dict(new);w,d,n=bad[0,0];bad[0,0]=(1-w,d,n)
    must_fail('flip a reachable write',bad,cert,mp)
    q,s=info['excluded_reads'][-1];c=deepcopy(cert);c['masks'][q]|=1<<(s<<cert['radius'])
    must_fail('admit a formerly impossible local window',new,c,mp)
    badmap=list(mp);badmap[0]=1
    must_fail('wrong start-state projection',new,cert,badmap)
    out={'old_states':len(old)//2,'new_states':len(new)//2,'state_window_pairs':info['state_window_pairs'],
         'closure_obligations':info['closure_obligations'],'projected_reachable_reads':len(info['reachable_read_overapproximation']),
         'excluded_reads':info['excluded_reads'],'negative_controls_rejected':rejected,
         'scope':'Inductive finite certificate; unknown neighboring cells always include both 0 and 1. Does not assume RH, source-loop invariants, or nonhalting.'}
    (ROOT/'reports/certificate_checks.json').write_text(json.dumps(out,indent=2)+'\n')
    return out
if __name__=='__main__':print(json.dumps(run(),indent=2))
\end{lstlisting}

\section{Control-transducer bisimulation}
\label{app:controlbisim}

The following Python source, \texttt{verify\_control.py}, interprets the
generated decision, dispatcher, and jump transitions literally. Register
leaves are replaced by their proved primitive contracts, and the cleanup
phase is replaced by its proved control effect.
Every possible conditional-decrement outcome is explored. The resulting finite
relation is checked for its initial pair, matching operation labels, and all
successor obligations.

\begin{lstlisting}[language=Python,basicstyle=\ttfamily\tiny]
"""Finite bisimulation of the two source-control transducers.

The generated dispatch/DAG/jump transitions are interpreted literally.
At register leaves, the existing register primitive contract supplies the
success/failure outcome. Every possible decrement outcome is explored;
increment always succeeds. No concrete counter range is assumed.
"""
from collections import deque
import json
from pathlib import Path

def encode(q,h,ones):return (q,h,tuple(sorted(ones)))
def seek(tab,c):
    q,h,positions=c;ones=set(positions);seen=set()
    while not q.startswith('3.') and q!='halt':
        key=encode(q,h,ones)
        if key in seen:raise ValueError('Control-only infinite loop')
        seen.add(key)
        if q=='2e.reg.cleanup_1':
            h-=1;q='6.continue.0'
        else:
            s=int(h in ones);w,d,q=tab[q,s]
            if w:ones.add(h)
            else:ones.discard(h)
            h+=1 if d=='R' else -1
        if not -5<=h<=100:raise ValueError('Control extent exceeded')
    return encode(q,h,ones)

def successor(tab,c,success):
    q,h,ones=c
    if not q.startswith('3.'):raise ValueError('Not an operation')
    if q.endswith('.inc') and not success:raise ValueError('Increment cannot fail')
    return seek(tab,('6.break.0' if success else '6.break.1',h-1,ones))

def generate(old,oldroot,new,newroot):
    start=(seek(old,('N'+str(oldroot),3,(0,1))),seek(new,('N'+str(newroot),3,(0,1))))
    ids={start:0};todo=deque([start]);rows=[]
    while todo:
        a,b=todo.popleft()
        if a[0]!=b[0]:raise ValueError(('Different operations',a,b))
        outputs=[]
        if a[0]!='halt':
            for s in ((1,) if a[0].endswith('.inc') else (0,1)):
                pair=successor(old,a,s),successor(new,b,s)
                if pair not in ids:ids[pair]=len(ids);todo.append(pair)
                outputs.append([s,ids[pair]])
        rows.append({'old':a,'new':b,'successors':outputs})
    return {'initial_pair':0,'pairs':rows,'semantics':'Register primitive outcomes; all possible decrement outcomes, increments always succeed.'}

def check(old,oldroot,new,newroot,cert):
    def c(v):return (v[0],v[1],tuple(v[2]))
    rows=cert['pairs'];assert cert['initial_pair']==0
    assert c(rows[0]['old'])==seek(old,('N'+str(oldroot),3,(0,1)))
    assert c(rows[0]['new'])==seek(new,('N'+str(newroot),3,(0,1)))
    edges=0
    for row in rows:
        a,b=c(row['old']),c(row['new']);assert a[0]==b[0]
        expected=[] if a[0]=='halt' else [1] if a[0].endswith('.inc') else [0,1]
        assert [s for s,_ in row['successors']]==expected
        for s,k in row['successors']:
            assert 0<=k<len(rows)
            assert successor(old,a,s)==c(rows[k]['old'])
            assert successor(new,b,s)==c(rows[k]['new']);edges+=1
    return {'operation_boundary_pairs':len(rows),'outcome_edges':edges,'halting_pairs':sum(r['old'][0]=='halt' for r in rows),'all_decrement_outcomes_checked':True}

if __name__=='__main__':
    from rh_source import RHBuilder as Old
    from rh120_source import RHBuilder as New
    a,b=Old(),New();old,new=a.build(False),b.build(False)
    cert=generate(old,a.root,new,b.root);info=check(old,a.root,new,b.root,cert)
    (Path(__file__).resolve().parent/'machines/control_bisimulation.json').write_text(json.dumps(cert,indent=2))
    (Path(__file__).resolve().parent/'reports/control_verification.json').write_text(json.dumps(info,indent=2))
    print(json.dumps(info,indent=2))
\end{lstlisting}

\section{Literal bootstrap replay checker}
\label{app:bootstrap}

The following standalone C++ program, \texttt{bootstrap\_check.cpp}, checks
the finite bootstrap statement \cref{lem:bootstrap} directly from the
generated named transition table. It does not import the source compiler
or any arithmetic routine.

\begin{lstlisting}[language=C++,basicstyle=\ttfamily\tiny,breaklines=true,columns=fullflexible,keepspaces=true]
#include <array>
#include <cassert>
#include <fstream>
#include <iostream>
#include <map>
#include <sstream>
#include <string>
#include <unordered_set>
#include <vector>

struct Transition { int write, dir, next; };
struct Raw { std::string q, next; int read, write; char dir; };

int main(int argc, char** argv) {
    const char* path = argc > 1 ? argv[1] : "machines/RH_121_raw_named.tm";
    std::ifstream in(path);
    if (!in) return 2;

    std::map<std::string,int> id;
    std::vector<std::string> name;
    std::vector<Raw> raw;
    std::string line, start;
    auto intern = [&](const std::string& s) {
        auto it=id.find(s); if (it!=id.end()) return it->second;
        int k=(int)id.size(); id[s]=k; name.push_back(s); return k;
    };

    while (std::getline(in,line)) {
        if (line.rfind("#! start",0)==0) {
            std::istringstream ss(line); std::string a,b; ss>>a>>b>>start;
            continue;
        }
        if (line.empty() || line[0]=='#') continue;
        Raw r; std::istringstream ss(line);
        ss>>r.q>>r.read>>r.write>>r.dir>>r.next;
        raw.push_back(r); intern(r.q); if (r.next!="H") intern(r.next);
    }
    std::vector<std::array<Transition,2>> table(id.size());
    for (const auto& r:raw)
        table[id[r.q]][r.read]={r.write,r.dir=='R'?1:-1,
                                r.next=="H"?-1:id[r.next]};

    std::unordered_set<long long> ones;
    long long head=0, steps=0;
    int state=id.at(start);
    while (name[state].empty() || name[state][0] != 'N') {
        int read=ones.count(head)?1:0;
        auto t=table[state][read];
        if (t.write) ones.insert(head); else ones.erase(head);
        head += t.dir; state=t.next; ++steps;
        assert(state>=0);
    }

    assert(steps==89775610LL);
    assert(name[state]=="N54");
    assert(head==-3818);
    std::unordered_set<long long> expected={-3821,-3820};
    for (long long p=3;p<=3821;p+=2) expected.insert(p);
    assert(ones==expected);
    std::cout << "bootstrap certificate verified\n";
}
\end{lstlisting}

\subsection{Literal completed-stage checker}

This C++ checker, \texttt{stage\_check.cpp}, detects returns from the main
source routine and validates the live-register tuple after completed
endpoints $2,3,4$.

\begin{lstlisting}[language=C++,basicstyle=\ttfamily\tiny]
#include <array>
#include <cassert>
#include <fstream>
#include <iostream>
#include <map>
#include <sstream>
#include <string>
#include <vector>
struct T {int w,d,n;};
struct Raw {std::string q,n; int s,w; char d;};
int main(int argc,char**argv){
    std::ifstream in(argc>1?argv[1]:"machines/RH_121_raw_named.tm");assert(in);
    std::map<std::string,int> id;std::vector<std::string> name;std::vector<Raw> rows;
    auto intern=[&](std::string s){auto it=id.find(s);if(it!=id.end())return it->second;
        int i=name.size();name.push_back(s);return id[s]=i;};
    std::string line;while(std::getline(in,line)){if(line.empty()||line[0]=='#')continue;
        Raw r;std::istringstream ss(line);ss>>r.q>>r.s>>r.w>>r.d>>r.n;rows.push_back(r);intern(r.q);intern(r.n);}
    std::vector<std::array<T,2>> tab(name.size());for(auto&x:tab)for(auto&t:x)t={0,1,-1};
    for(auto r:rows)tab[id[r.q]][r.s]={r.w,r.d=='R'?1:-1,id[r.n]};
    const int Z=1<<22;std::vector<unsigned char> tape(2*Z);int h=Z,q=id.at("0a.boot1.A"),stage=0;
    long long steps=0;bool first=true,returned=false;std::cout<<"[\n";
    while(stage<=3){
        assert(q>=0&&h>0&&h<2*Z-1&&steps<1000000000LL);
        bool entry=name[q]=="N54"&&(first||returned);
        if(entry){
            first=false;
            assert(h-Z==-3818&&tape[Z-3821]&&tape[Z-3820]);
            for(int p=0;p<Z-3821;++p)assert(!tape[p]);
            for(int p=Z-3819;p<Z+3;++p)assert(!tape[p]);
            assert(tape[Z+3]&&!tape[Z+4]);int p=Z+5;std::array<int,5> regs;
            for(int i=0;i<5;++i){int r=-1;while(tape[p]){++r;++p;}assert(r>=0);regs[i]=r;++p;}
            assert(regs[0]==stage&&regs[1]==0&&regs[2]==0&&regs[3]==0&&regs[4]==0);
            std::cout<<(stage?",\n":"")<<"  {\"completed_endpoint\":"<<(stage?stage+1:0)<<",\"steps\":"<<steps<<",\"m\":"<<regs[0]<<"}";
            ++stage;if(stage>3)break;
        }
        returned=name[q]=="5.root.2"&&tape[h]==1;
        auto t=tab[q][tape[h]];tape[h]=t.w;h+=t.d;q=t.n;++steps;
    }
    std::cout<<"\n]\n";
}
\end{lstlisting}

\section{Reproduction driver and mathematical checks}

\subsection{Reproduction driver}

The following Python source is \texttt{reproduce.py}.
\begin{lstlisting}[language=Python,basicstyle=\ttfamily\tiny,breaklines=true,columns=fullflexible,keepspaces=true]
"""Rebuild the 120-state machine and verify its finite certificates."""
from pathlib import Path
import argparse,json,hashlib
from rh120_source import RHBuilder
from rh_source import RHBuilder as Previous
from builder import save,count
from minimize import minimize
from tables import canonical,read,write
from window_certificate import generate,verify_projection
from compatible_merge import optimize,quot
from verify_control import generate as control_generate,check as control_check
from godel import encode,decode
ROOT=Path(__file__).resolve().parent

def dump(path,value):path.write_text(json.dumps(value,indent=2)+'\n')
def compile_named(b):
    raw=b.build(False);sm,mp,rounds=minimize(raw)
    for (q,s),(w,d,n) in raw.items():
        assert sm[mp[q],s]==(w,d,mp.get(n,'halt'))
    tab,ids,missing=canonical(sm);unmerged,rawids,_=canonical(raw)
    mapping=[ids[mp[q]] for q,i in sorted(rawids.items(),key=lambda x:x[1])]
    for (q,s),(w,d,n) in unmerged.items():
        assert tab[mapping[q],s]==(w,d,-1 if n<0 else mapping[n])
    return raw,tab,{q:ids[mp[q]] for q in mp},unmerged,mapping,rounds,missing

def build():
    for folder in ('machines','reports'):(ROOT/folder).mkdir(exist_ok=True)
    previous=Previous();oldraw,oldtab,*_=compile_named(previous)
    oldcert=generate(oldtab,2)
    from window_certificate import verify
    oldedges=verify(oldtab,oldcert)['reachable_read_overapproximation']
    oldgroups,_=optimize(oldtab,oldedges);oldfinal,oldmap=quot(oldtab,oldedges,oldgroups)
    verify_projection(oldtab,oldfinal,oldcert,oldmap)
    assert oldfinal==read(ROOT/'baseline/RH_121.tm')
    b=RHBuilder();raw,tab,names,unmerged,bmap,rounds,missing=compile_named(b)
    cert=generate(tab,2);info=verify(tab,cert)
    groups,trace=optimize(tab,info['reachable_read_overapproximation'])
    final,projection=quot(tab,info['reachable_read_overapproximation'],groups)
    verify_projection(tab,final,cert,projection)
    assert len(tab)==242 and len(final)==240
    expected=ROOT/'machines/RH_120.tm'
    if expected.exists():assert read(expected)==final
    save(raw,ROOT/'machines/RH_121_raw_named.tm')
    for name,t in [('RH_121_predecessor',tab),('RH_122_unmerged',unmerged),('RH_120',final)]:
        write(t,ROOT/'machines'/f'{name}.tm')
    dump(ROOT/'machines/bisimulation.json',bmap)
    dump(ROOT/'machines/window_certificate.json',cert)
    dump(ROOT/'machines/RH_120.projection.json',projection)
    dump(ROOT/'machines/RH_121.dispatch.json',b.dispatch_certificate)
    dump(ROOT/'machines/RH_121.names.json',names)
    dump(ROOT/'machines/RH_120.names.json',{q:projection[i] for q,i in names.items()})
    control=control_generate(oldraw,previous.root,raw,b.root)
    control_info=control_check(oldraw,previous.root,raw,b.root,control)
    dump(ROOT/'machines/control_bisimulation.json',control)
    for name,t in [('RH_121_predecessor',tab),('RH_120',final)]:
        n,g,k=encode(t);assert decode(n,g)==t
        dump(ROOT/'machines'/f'{name}.goedel.json',{'states':n,'G_bits':g.bit_length(),'G_decimal':str(g),'K_decimal':str(k)})
    # Deliberate corruptions must be rejected by the projection/control checks.
    bad=dict(final);w,d,n=bad[0,0];bad[0,0]=(1-w,d,n)
    try:verify_projection(tab,bad,cert,projection)
    except ValueError:pass
    else:raise AssertionError('Corrupted reachable write accepted')
    changed=json.loads(json.dumps(control));changed['pairs'][0]['new'][0]='halt'
    try:control_check(oldraw,previous.root,raw,b.root,changed)
    except AssertionError:pass
    else:raise AssertionError('Corrupted control pair accepted')
    info.pop('reachable_read_overapproximation')
    report={'original_121_reproduced_exactly':True,'raw_counts':count(raw),'predecessor_states':len(tab)//2,
        'final_states':len(final)//2,'transitions':len(final),'minimizer_rounds':rounds,'completed_missing':missing,
        'window':info,'merge_trace':trace,'control_bisimulation':control_info,'G_bits':g.bit_length(),
        'sha256':hashlib.sha256((ROOT/'machines/RH_120.tm').read_bytes()).hexdigest(),
        'negative_controls_rejected':['reachable write mutation','control-operation mismatch']}
    dump(ROOT/'reports/build.json',report);return report

if __name__=='__main__':
    p=argparse.ArgumentParser();p.add_argument('--test',action='store_true');args=p.parse_args()
    print(json.dumps(build(),indent=2))
    if args.test:
        from verify_math import run
        print(json.dumps(run(),indent=2))
\end{lstlisting}
\subsection{Exact host-integer checks}

The following Python source is \texttt{verify\_math.py}.
\begin{lstlisting}[language=Python,basicstyle=\ttfamily\tiny,breaklines=true,columns=fullflexible,keepspaces=true]
"""Independent exact host-integer checks; these are NOT binary-tape executions."""
from pathlib import Path
from fractions import Fraction
from itertools import product
import json,math
ROOT=Path(__file__).resolve().parent

def ell(B,b):
    assert B>=0 and b>=2
    d=0
    while B: B//=b; d+=1
    return d

def beta(B,b):
    assert B>=0 and b>=2
    d=0
    while B: B=(B-1)//b; d+=1
    return d

def fold_unit(x,j,v):
    while j:j-=1;v+=1;x+=1
    while v:v-=1;j+=1
    return x,j,v

def transform_unit(x,j,v):
    m=x;x=0
    while True:
        x,j,v=fold_unit(x,j,v)
        if not m:break
        m-=1;x+=1
    return x,j,v,m

def divide_unit(x,j,v):
    m=x;x=0;v+=j;j=0
    while m:
        m-=1
        if v:v-=1;j+=1
        else:v+=j;j=0;x+=1
    return x,j,v,m

def prime_flags(N):
    s=bytearray(b'\1')*(N+1);s[0:2]=b'\0\0'
    for p in range(2,math.isqrt(N)+1):
        if s[p]:s[p*p:N+1:p]=b'\0'*((N-p*p)//p+1)
    return s

def scaling(j,f):
    m=q=0
    while j:
        j-=1;m+=1
        for _ in range(2):m+=(j+1)//2;j//=2
        f//=2;q+=1
    return m,f,q

def scanner_checks():
    from rh120_source import RHBuilder
    total=0;reports={}
    for cls,name in [(RHBuilder,'RH_121')]:
        b=cls(True);tab=b.build(False);cert=b.dispatch_certificate;k=cert['scanner_states']
        assert k==4
        stored=json.loads((ROOT/'machines'/f'{name}.dispatch.json').read_text())
        assert stored==json.loads(json.dumps(cert))
        words={r['bits'] for r in cert['control_words']}
        # Also cover every synthetic length <=10 word beginning in 1 and
        # satisfying the stated forbidden-run premise.
        for length in range(1,11):
            for suffix in product('01',repeat=length-1):
                w='1'+''.join(suffix)
                if '0'*k not in w:words.add(w)
        cases=0
        for word in words:
            assert '0'*k not in word and word[0]=='1'
            tape={i:int(s) for i,s in enumerate(word)}
            for start,bit in enumerate(word):
                if bit!='1':continue
                h=start-1;q='4.run.0';steps=0;copy=dict(tape)
                while q!='5.root.1':
                    assert steps<10*(len(word)+k)+10
                    read=copy.get(h,0);w,d,q=tab[q,read]
                    assert w==read # scanner is read-only
                    h+=1 if d=='R' else -1;steps+=1
                assert h==1 and copy==tape
                cases+=1
        total+=cases;reports[name]={'scanner_states':k,'maximum_zero_run':cert['maximum_consecutive_zeros'],
                'certified_leaf_paths':cert['terminal_paths'],'test_words':len(words),'read_only_scan_executions':cases}
    return {'cases':total,'machines':reports,'scope':'Finite language certificates and synthetic operation-entry scanner runs; inherited tape-layout invariant is an additional semantic premise.'}

def run():
    counts={};c=0
    for x,j,v in product(range(7),repeat=3):
        assert fold_unit(x,j,v)==(x+j,j+v,0)
        assert transform_unit(x,j,v)==(x*(j+v+1)+j,j+v,0,0);c+=1
    counts['fold_unit_cases']=counts['transform_unit_cases']=c;c=0
    for B in range(40):
        for j,v in product(range(7),repeat=2):
            J=j+v+1;q,r=divmod(B,J)
            got=divide_unit(B,j,v);assert got==(q,r,J-1-r,0)
            assert transform_unit(*got[:3])==(B,J-1,0,0);c+=1
    counts['split_divider_and_reconstruction_cases']=c;c=0
    for b in range(2,18):
        for B in list(range(257))+[b**d-1 for d in range(1,20)]+[b**d for d in range(1,20)]+[sum(b**i for i in range(1,d+1)) for d in range(1,15)]:
            e=ell(B,b);d=beta(B,b)
            assert max(e-1,0)<=d<=e
            if B:assert sum(b**i for i in range(1,d))<B<=sum(b**i for i in range(1,d+1))
            if e:assert b**(e-1)<=B<b**e
            c+=1
    counts['offset_digit_witnesses']=c
    for n in range(1,100001):
        j=n-1;q=0
        while j:j=(j-1)//4;q+=1
        assert q==((3*n).bit_length()-1)//2
    counts['scaling_clock_cases']=100000;c=0
    for j in range(1000):
        q=((3*(j+1)).bit_length()-1)//2
        W=(q+1)*2**q
        for f in (0,1,2,3,7,15,W-1,W,W+1,2*W):
            assert scaling(j,f)==(j,f//2**q,q);c+=1
    counts['unit_conserving_scaling_cases']=c
    small=[(1,5,Fraction(4,3)),(2,21,Fraction(8,3)),(3,85,Fraction(4)),(4,341,Fraction(16,3)),
           (5,1365,10*Fraction(56,81)+Fraction(682,2389)),(6,2656,Fraction(22,3))]
    assert all(Fraction(N)/L<(q+1)*2**q for q,N,L in small)
    counts['small_range_rational_inequalities']=len(small)
    sieve=prime_flags(4096);A=F=P=1
    for j in range(2,4097):
        assert (A%j!=0)==bool(sieve[j]);A*=j*(j if A%j else 1)
        F*=j
        if sieve[j]:P*=j
        assert A==F*P
    counts['self_sieve_identities']=4095
    rows=[];steps=0;flags=prime_flags(256)
    for n in range(2,257):
        # Execute the source-level product control through exact unit-kernel
        # contracts, with an independent sieve/factorial oracle at EVERY step.
        x=j=0;f=n-2;F=P=1
        while True:
            j+=1;J=j+1;q,r=divmod(x,J);v=J-1-r;x=q;j=r
            if v:x,j,v=x*(j+v+1)+j,j+v,0
            for _ in range(2):x,j,v=x*(j+v+1)+j,j+v,0
            F*=J
            if flags[J]:P*=J
            assert x==F*P-1 and j==J-1 and v==0;steps+=1
            if not f:break
            f-=1
        f=j;v=j;j=0;old=x;division_steps=0
        while x:
            x-=1;J=j+v+1;x,j=divmod(x,J);v=J-1-j
            f=max(f-1,0);division_steps+=1
        j+=v;v=0
        assert division_steps==beta(old,n)
        assert j==n-1 and f==max(n-1-beta(F*P-1,n),0)
        delta=f;m,f,q=scaling(j,f);final=max(f-q,0)
        assert m==n-1 and final==0 # finite validation, not an RH proof
        rows.append({'n':n,'ordinary_digits':ell(old,n),'offset_digits':division_steps,'deficit':delta,'clock':q,'threshold':(q+1)*2**q,'post_compare_deficit':final})
    counts['complete_host_integer_stages']=len(rows);counts['intermediate_product_invariants']=steps
    report={'counts':counts,'scanner':scanner_checks(),'stages':rows,
            'scope':'Exact host-integer and isolated scanner checks. Full arithmetic stages use mathematically specified kernel contracts, not binary-tape execution or AST formal verification.'}
    (ROOT/'reports/math_checks.json').write_text(json.dumps(report,indent=2)+'\n')
    return {k:v for k,v in report.items() if k!='stages'}
if __name__=='__main__':print(json.dumps(run(),indent=2))
\end{lstlisting}
\subsection{Reversible G\"odel coding}

The following Python source is \texttt{godel.py}.
\begin{lstlisting}[language=Python,basicstyle=\ttfamily\scriptsize,breaklines=true,columns=fullflexible,keepspaces=true]
"""Exact reversible coding of a complete binary transition table."""
import sys
if hasattr(sys, "set_int_max_str_digits"): sys.set_int_max_str_digits(20000)

def encode(table):
    n=len(table)//2;b=4*(n+1)
    G=0
    for q in reversed(range(n)):
        for s in (1,0):
            w,move,target=table[q,s]
            digit=w+2*(move==1)+4*(n if target==-1 else target)
            G=G*b+digit
    return n,G,(1<<n)*(2*G+1)

def decode(n,G):
    b=4*(n+1);table={}
    for q in range(n):
        for s in (0,1):
            G,v=divmod(G,b);target=v//4
            table[q,s]=(v%2,1 if (v//2)%2 else -1,-1 if target==n else target)
    assert G==0
    return table

\end{lstlisting}

\section{Exact 120-state transition table}
\label{app:table}

The syntax is \texttt{state read write direction target}; \texttt{H} denotes
the uncounted halting target. The table is the text file
\texttt{machines/RH\_120.tm}.
\begin{lstlisting}[language={},basicstyle=\ttfamily\tiny,breaklines=false,columns=fullflexible,keepspaces=true]
# start 0; blank 0; halt H; one tape; two symbols; L/R only
0 0 1 R 1
0 1 1 L 2
1 0 0 L 0
1 1 0 L 3
2 0 1 L 0
2 1 1 R 4
3 0 1 L 1
3 1 1 R 5
4 0 1 R 6
4 1 0 R 7
5 0 0 R 3
5 1 0 R 1
6 0 0 R 8
6 1 1 R H
7 0 0 R 9
7 1 1 R 7
8 0 1 R 8
8 1 1 L 10
9 0 0 R 11
9 1 1 R 9
10 0 0 L 10
10 1 1 L 12
11 0 1 R 13
11 1 1 R 11
12 0 0 L 14
12 1 1 L 12
13 0 0 L 15
13 1 0 R 11
14 0 0 L 16
14 1 1 L 12
15 0 0 L 17
15 1 1 L 15
16 0 0 L 18
16 1 1 L 12
17 0 0 L 19
17 1 1 L 15
18 0 0 R 20
18 1 1 L 12
19 0 1 R 21
19 1 0 L 19
20 0 0 R 20
20 1 1 R 22
21 0 0 R 21
21 1 0 R 23
22 0 0 R 24
22 1 1 R 25
23 0 0 L 24
23 1 0 R 23
24 0 1 L 10
24 1 1 R 26
25 0 0 R 27
25 1 0 R 27
26 0 0 R 24
26 1 1 R 28
27 0 0 R 29
27 1 1 R 30
28 0 0 R 4
28 1 1 R 31
29 0 0 R 31
29 1 1 R 32
30 0 0 R 33
30 1 1 R 34
31 0 1 R 6
31 1 0 R 35
32 0 0 R 36
32 1 0 L 37
33 0 0 R 38
33 1 1 R 39
34 0 0 R 40
34 1 1 R 41
35 0 0 R 42
35 1 1 R 35
36 0 1 R 6
36 1 0 R 43
37 0 0 L 37
37 1 0 L 10
38 0 0 R 44
38 1 1 R 45
39 0 0 R 46
39 1 0 L 37
40 0 0 R 47
40 1 1 R 48
41 0 0 R 49
41 1 1 R 50
42 0 0 R 51
42 1 1 R 42
43 0 0 R 52
43 1 1 R 43
44 0 1 R 6
44 1 0 R 53
45 0 0 R 54
45 1 1 R 55
46 0 1 R 6
46 1 0 R 56
47 0 1 R 6
47 1 0 R 57
48 0 0 R 58
48 1 0 L 37
49 0 0 R 59
49 1 1 R 60
50 0 0 R 61
50 1 1 R 62
51 0 1 L 63
51 1 0 R 64
52 0 0 R 53
52 1 1 R 52
53 0 0 R 7
53 1 1 R 53
54 0 0 R 65
54 1 1 R 66
55 0 0 R 67
55 1 1 R 68
56 0 0 R 69
56 1 1 R 56
57 0 0 R 35
57 1 1 R 57
58 0 0 R 70
58 1 1 R 36
59 0 0 R 71
59 1 1 R 72
60 0 0 R 73
60 1 1 R 74
61 0 0 R 47
61 1 1 R 75
62 0 0 R 76
62 1 1 R 77
63 0 0 L 78
63 1 1 L 63
64 0 0 L 51
64 1 1 R 79
65 0 0 R 47
65 1 1 R 80
66 0 0 R 81
66 1 1 R 82
67 0 0 R 73
67 1 1 R 83
68 0 0 R 84
68 1 1 R 84
69 0 0 R 57
69 1 1 R 69
70 0 1 R 6
70 1 0 R 85
71 0 1 R 6
71 1 0 R 69
72 0 0 R 54
72 1 1 R 86
73 0 1 R 6
73 1 0 R 87
74 0 0 R 44
74 1 0 L 37
75 0 0 R 88
75 1 1 R 89
76 0 0 R 71
76 1 1 R 90
77 0 0 R 46
77 1 1 R H
78 0 0 L 91
78 1 1 L 63
79 0 0 R 92
79 1 1 R 79
80 0 0 R 70
80 1 0 L 37
81 0 0 R 71
81 1 1 R 93
82 0 0 R 31
82 1 1 R 94
83 0 0 R 70
83 1 1 R 84
84 0 0 R 81
84 1 1 R 95
85 0 0 R 43
85 1 1 R 85
86 0 0 R 96
86 1 0 L 32
87 0 0 R 56
87 1 1 R 87
88 0 0 R 4
88 1 1 R 97
89 0 0 R 98
89 1 0 L 32
90 0 0 R 96
90 1 0 L 37
91 0 0 L 19
91 1 0 L 91
92 0 0 L 99
92 1 1 R 79
93 0 0 R 4
93 1 0 L 37
94 0 0 R 100
94 1 0 L 37
95 0 0 R 101
95 1 1 R 102
96 0 0 R 46
96 1 0 L 19
97 0 0 R 103
97 1 1 R 104
98 0 1 R 6
98 1 0 R 52
99 0 1 L 105
99 1 1 L 99
100 0 0 R 106
100 1 1 R 107
101 0 0 R 108
101 1 1 R 60
102 0 0 R 31
102 1 1 R 89
103 0 0 R 109
103 1 1 R 60
104 0 0 R 103
104 1 1 R 110
105 0 0 L 15
105 1 0 L 99
106 0 0 R 73
106 1 1 R 111
107 0 0 R 65
107 1 1 R 98
108 0 0 R 47
108 1 1 R 112
109 0 0 R 47
109 1 1 R 113
110 0 0 R 114
110 1 1 R 115
111 0 0 R 44
111 1 0 L 91
112 0 0 R 70
112 1 1 R 89
113 0 0 R 116
113 1 0 L 37
114 0 0 R 46
114 1 1 R 117
115 0 0 R 73
115 1 1 R 32
116 0 0 R 4
116 1 1 R 118
117 0 0 R 119
117 1 0 L 37
118 0 0 R 47
118 1 1 R 70
119 0 0 R 46
119 1 1 R 70
\end{lstlisting}

\clearpage

\paragraph{AI assistance.}
The idea for the Turing machine is the author's. GPT~6 Astra assisted
in writing the proofs and the Turing-machine code, as well as in
manuscript preparation.

\small
\setlength{\bibsep}{0.25em}

\end{document}